\documentclass[journal,onecolumn, 12pt]{IEEEtran}

\usepackage[utf8]{inputenc}

\usepackage{amsthm}
\usepackage[T1]{fontenc}
\usepackage{url}
\usepackage{amssymb}
\usepackage{comment}
\usepackage{color}
\usepackage{float}
\usepackage{mathrsfs}
\usepackage{balance}
\usepackage[toc, page]{appendix}
\usepackage{mathtools}
\usepackage{ifthen}
\usepackage{cite}

\begin{document}
\newcommand{\s}[1]{\vskip #1 in}
\newcommand{\n}{\noindent}
\newcommand{\p}{\pmb{\lambda}}
\newcommand{\gd}[1]{\textbf{\textcolor{red}{[gd:#1]}}}
\newcommand{\pk}[1]{\mathbf{\textcolor{blue}{[#1]}}}
\newcommand{\Hr}{H_\mathbb{R}}
\newtheorem{theorem}{Theorem}
\newtheorem*{theorem*}{Theorem}
\newtheorem*{lemma*}{Lemma}
\newtheorem{corollary}{Corollary}
\newtheorem{lemma}{Lemma}
\newtheorem{definition}{Definition}
\newtheorem{prop}{Proposition}
\newcommand{\nn}{\nonumber}
\newcommand{\myvec}[1]{\ensuremath{\begin{bmatrix}#1\end{bmatrix}}}
\def\an#1{{\color{blue}#1}}
\title{On the Sample Complexity and Optimization Landscape for Quadratic Feasibility Problems} 


\author{%
   \IEEEauthorblockN{Parth K.~Thaker, Gautam Dasarathy, and Angelia Nedi\'c}\\
   \IEEEauthorblockA{Electrical, Computer and Energy Engineering \\ Arizona State University, 
                     Tempe, USA\\
                     Email: \{pkthaker, gautamd, angelia.nedich\}@asu.edu}
}

\maketitle

\begin{abstract}
   We consider the problem of recovering a complex vector $\mathbf{x}\in \mathbb{C}^n$ from $m$ quadratic measurements $\{\langle A_i\mathbf{x}, \mathbf{x}\rangle\}_{i=1}^m$. This problem, known as quadratic feasibility, encompasses the well known phase retrieval problem and has applications in a wide range of important areas including power system state estimation and x-ray crystallography. In general, not only is the the quadratic feasibility problem NP-hard to solve, but it may in fact be unidentifiable. In this paper, we establish conditions under which this problem becomes {identifiable}, and further prove isometry properties in the case when the matrices $\{A_i\}_{i=1}^m$ are Hermitian matrices sampled from a complex Gaussian distribution. Moreover, we explore a nonconvex {optimization} formulation of this problem, and establish salient features of the associated optimization landscape that enables gradient algorithms with an arbitrary initialization to converge to a \emph{globally optimal} point with a high probability. Our results also reveal sample complexity requirements for successfully identifying a feasible solution in these contexts. 

\end{abstract}


\section{Introduction}
\let\thefootnote\relax\footnote{This work is supported in part by the National Science Foundation under Grant No. OAC-1934766 and CCF-1717391}
Finding a solution to a system of quadratic equations is an important  problem with a wide range of applications. 
{It}~arises in areas such as  power system state estimation \cite{wang2016power}, phase retrieval \cite{candes2015phase, candes2013phaselift, eldar2014phase, balan2014phase}, x-ray crystallography \cite{drenth2007principles}, the turnpike problem~\cite{dakic2000turnpike}, and unlabeled distance geometry problems\cite{duxbury2016unassigned, huang2018reconstructing} among others. 
Such problems are often reduced to a {\it quadratic feasibility problem},
where one is concerned with
finding a \emph{feasible} vector $\mathbf{x}$ that conforms to a set of quadratic observations of the form $\{\langle A_i\mathbf{x}, \mathbf{x}\rangle\}_{i=1}^m$ with respect to a set $\left\{ A_i \right\}_{i=1}^m$
{of measurement matrices}. {Formally}, it can be cast as:
\begin{align}\tag{P1}\label{eq:main_prob}
\text{find }\mathbf{x}
\text{ \ \ such that} &\ \langle A_i\mathbf{x}, \mathbf{x}\rangle = c_i, ~~~ \forall i=1,2,\dots,m.
\end{align}

The quadratic feasibility problem is an instance of quadratically constrained quadratic programs (QCQPs)\cite{park2017general}, which has enjoyed a long and rich  research history dating back to 1941\cite{dines1941mapping}. Given their broad applicability to critical problems, research in QCQPs continues to be of active interest \cite{park2017general, beck2017branch, beck2009convexity, beck2006strong}. Unfortunately, it is known that solving QCQPs is an NP-hard problem \cite{sahni1974computationally}. This combined with the lack of tractable duality properties \cite{polik2007survey} has made it hard to establish a sound theoretical framework for understanding the solutions and computing them. 
However, an extremely productive line of research has instead considered subclasses of QCQPs 
that are both practically relevant and can be analyzed.
In this paper, we take a similar approach and identify an important subclass of QCQPs that have a broad range of applications. In particular, we analyze the quadratic feasibility problem, and establish conditions under which such problems are identifiable, and then show that these conditions are in-fact sufficient for the efficient computation of their solutions. 

We start by considering quadratic functions $\{\langle A_i \mathbf{x}, \mathbf{x}\rangle\}_{i=1}^m$, where $\mathbf{x}\in \mathbb{C}^n$ and $A_i\in \mathbb{C}^{n\times n}$ are Hermitian {matrices}. 
We focus on their ability to generate injective maps up-to a phase factor (note that the quadratic functions $\{\langle A_i \mathbf{x}, \mathbf{x}\rangle\}_{i=1}^m$ are invariant to phase shifts). We establish a relationship between injectivity and isometry and show that, in real world scenarios, it is not difficult for a set of quadratic measurements $\{\langle A_i \mathbf{x}, \mathbf{x}\rangle\}_{i=1}^m$ to possess such an isometry property by establishing that this holds with very high probability when the matrices $\{A_i\}_{i=1}^m$ are complex Gaussian random matrices.

After establishing injectivity (and hence the identifiability) of the problem, we consider the question of computationally tractable approaches to actually find a feasible solution. Toward this end, a natural approach is to optimize the appropriate $\ell_2$-loss. Unfortunately, this turns out to be a nonconvex problem, that is NP-hard to solve in general. However, we show that under the same conditions required to establish injectivity, the landscape of this optimization problem is well-behaved. This allows us to establish that any gradient based algorithm with an arbitrary initialization can recover a globally optimal solution almost surely. 

The rest of the paper is organized as follows. Section \ref{sec:main_results} highlights the main results of this work. We discuss some related work in Section~\ref{sec:related_works}. In Section~\ref{sec:injectivity} we establish and analyze isometry/identifiability properties of the mapping $\{\langle A_i\mathbf{x}, \mathbf{x}\rangle\}$ when the measurement matrices are complex 
Gaussian. Finally, Section \ref{sec:Opt_landscape} casts the problem as a quadratic loss minimization problem (suitable for efficient algorithms) and establishes favorable properties of the loss landscape that allows one to find \an{a} solution using 
gradient-based methods with {arbitrary initial points}. 



\section{Main Results}\label{sec:main_results}
Before we state our main results of the paper, we introduce some notation that will be used throughout the paper. 
\subsection{Notation}
For any $r\in \mathbb{N}$, 
we write $[r]$ to denote the set $\{1,2,\dots, r\}$.
We let $\mathbb{C}^n$ and $\mathbb{R}^n$  denote the $n$-dimensional complex and real vector spaces, respectively. Unless otherwise stated, bold letters such as $\mathbf{x}$ indicate vectors in $\mathbb{C}^n$; $\mathbf{x}_{\mathbb{R}}$ and $\mathbf{x}_{\mathbb{C}}$ denote the real and the imaginary part of the vector $\mathbf{x}$, respectively. We denote complex conjugate of  $\mathbf{x}$ by $\bar{\mathbf{x}}$. 
Capital letters such as $X$ denote matrices in $\mathbb{C}^{n\times n}$. 
The use of $\mathsf{i}$ (without serif) indicates the complex square root of -1 (we will use $i$ to indicate an indexing variable).
We let $\mathcal{S}^{a,b}(\mathbb{R}^{n\times n})$ denote the set of all matrices $X\in \mathbb{R}^{n\times n}$ having $a$ non-negative eigenvalues and $b$ negative eigenvalues, where $a+b=n$. 
The set $\mathbf{H}_n(\mathbb{C})$ denotes the set of all $n\times n$ Hermitian matrices. 
We write $A^\top$ and $A^*$ to denote, respectively,
the transpose and the Hermitian transpose (transpose conjugate) of a matrix $A$. 
We use $\langle\cdot, \cdot\rangle$ to denote the inner vector product in the complex space. The symmetric outer product, denoted by  $[[\cdot,\cdot]]$, is defined as 
    $[[\mathbf{u}, \mathbf{v}]] = \mathbf{u}\mathbf{v}^* + \mathbf{v}\mathbf{u}^*.$
Finally, we will let $\sim$ denote the following equivalence relation on $\mathbb{C}^n$: 
$\mathbf{x}\sim\mathbf{y}$ if and only if 
$\mathbf{x}=c\mathbf{y}$ for some $c\in \mathbb{C}$
with $|c| = 1$. 
We will write $\widehat{\mathbb{C}}^n\triangleq\mathbb{C}^n/\sim$ to denote the associated quotient space.
Given a set of matrices 
{$\mathcal{A} = \{A_i\}_{i=1}^m\subset\mathbf{H}_n(\mathbb{C})$}, we will let $\mathcal{M}_{\mathcal{A}}$ denote the following mapping from $\widehat{\mathbb{C}}^n\to \mathbb{C}^m$:
\begin{equation}\label{eq:defmapping}
    \mathcal{M}_{\mathcal{A}}(\mathbf{x}) = (\langle A_1\mathbf{x}, \mathbf{x}\rangle, \langle A_2\mathbf{x}, \mathbf{x}\rangle,\dots, \langle A_m\mathbf{x}, \mathbf{x}\rangle).
\end{equation}
While $\mathcal{M}_{\mathcal{A}}$ technically operates on the equivalence classes in $\widehat{\mathbb{C}}^n$, we will abuse the notation slightly and think of $\mathcal{M}_\mathcal{A}$ as operating on the elements of $\mathbb{C}^n$. 
\subsection{Main Results}
We consider the quadratic feasibility 
problem~\eqref{eq:main_prob}
with the complex decision vector, i.e., 
$\mathbf{x}\in \mathbb{C}^n$,
Hermitian matrices $A_i\in \mathbb{C}^{n\times n}$ 
and real numbers $c_i\in \mathbb{R}$ for all $i\in [m]$. In order to understand the properties of this problem, we need a coherent distance metric that is analytically tractable. Toward this, we will use the following: 
\begin{equation}\label{eq:defdist}
    d(\mathbf{x}, \mathbf{y}) = \|\mathbf{x}\mathbf{x}^* - \mathbf{y}\mathbf{y}^*\|_F
    \quad\hbox{for any } \mathbf{x},\mathbf{y}\in \mathbb{C}^n.
\end{equation}
Notice that this distance metric is invariant under phase shifts, 
and has been used to prove crucial robustness results for the phase retrieval problem (however, only in $\mathbb{R}^n$); see e.g., \cite{bandeira2014saving, balan2015invertibility,candes2013phaselift}. 

Our first main result (Theorem~\ref{thm:isometry}) 
{states that the mapping $\mathcal{M}_\mathcal{A}$ is a near-isometry
when the matrices $A_i$ are chosen 
from a complex Gaussian distribution.
A sketch of the statement is provided below.}

\begin{theorem}[sketch]\label{thm:thm1-sketch}
Let $\mathcal{A} = \{A_i\}_{i=1}^m$ be a set of complex Gaussian random matrices. 
{Suppose that the number $m$ of measurements 
satisfies $m > Cn$, for a large enough $C>0$.
Then, with a high probability, the following relation
holds: for some constants $\alpha, \beta>0$, 
and for all $\mathbf{x},\mathbf{y}\in \mathbb{C}^n$,
\[    \alpha d(\mathbf{x}, \mathbf{y})\leq\|\mathcal{M}_\mathcal{A}(\mathbf{x}) - \mathcal{M}_\mathcal{A}(\mathbf{y})\|_2 \leq \beta d(\mathbf{x}, \mathbf{y}).
\]
}
\end{theorem}
In other words, $\mathcal{M}_\mathcal{A}$ nearly preserves distances with respect to the distance measure defined in \eqref{eq:defdist}. Therefore, when $\mathbf{x}$ and $\mathbf{y}$ are distinct, the corresponding measurements are also distinct -- that is, the measurement model defined by $\mathcal{M}_\mathcal{A}$ is \emph{identifiable}. The formal statement along with the full proof is presented in Theorem~\ref{thm:isometry} in Section~\ref{sec:injectivity}. 

Having established that \eqref{eq:main_prob} has a uniquely identifiable solution (upto a phase ambiguity), we next turn our attention to finding a feasible solution in a computationally efficient manner. 
To this end, one may consider recasting the quadratic feasibility problem as a quadratic {minimization problem of the $\ell_2$-loss function, 
as follows:
\begin{equation}\tag{P2}\label{eq:ell2loss0}
    \min_{\mathbf{x}\in\mathbb{C}^n} 
    f(\mathbf{x}),\qquad
    f(\mathbf{x})\triangleq\frac{1}{m}\sum_{i=1}^m\left|\langle A_i\mathbf{x}, \mathbf{x}\rangle - c_i\right|^2.
\end{equation}
}
Unfortunately, this optimization problem is non-convex and, in general, one may not expect any 
{gradient based method to converge to a global minimum.} However, our next main result 
{states that with high probability, 
the optimization landscape of $\eqref{eq:ell2loss0}$ is in fact amenable to gradient based methods!} 
Moreover, we {establish this} result under the same condition required for the problem to be identifiable -- namely, the measurement matrices are drawn from the complex Gaussian distribution. 
We now provide a sketch of our second main result. 
\begin{theorem}[sketch]\label{thm:thm2-sketch}
Let $\mathcal{A}= \{A_i\}_{i=1}^m$ be a set of complex Gaussian random matrices and let $\mathbf{z}$ be a global optimizer of \eqref{eq:ell2loss0}. Then, with high probability, the following holds:
\begin{enumerate}
    \item 
    $\mathcal{M}_{\mathcal{A}}(\mathbf{w}) = \mathcal{M}_{\mathcal{A}}(\mathbf{z})$
    for all local minima $\mathbf{w}$ 
    of \eqref{eq:ell2loss0}.
    \item The function $f$ in \eqref{eq:ell2loss0} has the strict saddle property.
\end{enumerate}
\end{theorem}
This theorem states that provided the measurement matrices are Gaussian, with high probability, 
the {minimization} problem \eqref{eq:ell2loss0} has no spurious local minima, and any saddle point of the {function $f$ is \emph{strict} in the sense that $f$ has  a strictly negative curvature at such a point.} The latter property, called the strict saddle property, is defined in Section~\ref{sec:Opt_landscape}, where we also provide a formal statement of  Theorem~\ref{thm:strictsaddle}. 
Finally, {based on the properties established here about the loss landscape, we establish that a solution to problem~\eqref{eq:ell2loss0} can be obtained
by applying a gradient based algorithm (from an arbitrary initial point), since such an algorithm is unlikely to converge to a saddle point. 
}

\setcounter{theorem}{0}
\section{Related work}\label{sec:related_works}

QCQPs have enjoyed a lot of attention over the last century. However, due to the limitation of the duality properties of QCQPs~\cite{barvinok1995problems}, a significant fraction of research has focused predominantly on heuristic approaches to their solution~\cite{konar2017fast, konar2017first, konar2017non}. Recently, an ADMM-based method 
has been proposed in\cite{huang2016consensus} 
with an asymptotic convergence result based on the duality properties QCQPs. {Our results in this paper
bring new insights to this area by analyzing a subset of QCQPs, namely, the quadratic feasibility problems}.



The quadratic feasibility problem~\eqref{eq:main_prob} 
arises in many applications, including phase retrieval \cite{candes2015phase} and power system state estimation \cite{wang2016power}. Phase retrieval in and of itself finds applications in a wide variety of fields such as imaging, optics, quantum tomography, audio signal processing with a wide literature, including~\cite{eldar2014phase, candes2015phase, tan2017phase, balan2014phase}. 
In~\cite{candes2013phaselift}, an 
approximate $\ell_1$ isometry property was 
established for the phase retrieval problem, but the bounds therein are not strong enough to provide RIP-like guarantees. In this paper, we improve these bounds to establish isometry results for a large class of problems and provide RIP-type bounds 
for the case when $\{A_i\}_{i=1}^m$ are complex Gaussian random matrices. 


A feasibility problem is often cast as a minimization problem with a suitably chosen loss function. 
Even with a nonconvex objective, gradient based methods have proven to work for phase retrieval \cite{tan2017phase, balan2014phase, candes2015phase}, matrix factorization \cite{bhojanapalli2016dropping, jin2016provable} and robust linear regression \cite{jain2017non}. The work in \cite{sun2018geometric} has established landscape properties for the phase retrieval problem, which sheds light on the success of gradient based methods in solving the problem. 
In this work, we extend these results to a 
wider class of problems along with additional
insights into the problem properties. 
In~\cite{ge2017no}, it was shown that many nonconvex loss functions have specific landscape properties, which allows gradient based algorithm to recover a globally optimal solution without any additional information. One unfortunately cannot readily transport those results to our setting, {mainly} due to the significant differences between the real and complex vector spaces.  
For instance, a quadratic feasibility problem 
in $\mathbb{R}^n$ has only two isolated local minima, 
{while it has a continuum} of minima
in $\mathbb{C}^n$.

{The authors in \cite{wang2017generalized} } provide lower bounds on the minimum number of independent measurements required for a successful recovery for the quadratic feasibility problem. 
More recently, \cite{huang2019solving} showed that 
the quadratic feasibility problem 
can be solved, with high probability, by gradient descent provided a good initialization is used. In contrast, the current work 
takes a parallel track by analyzing the landscape of the associated 
$\ell_2$-loss function. In particular,
for the $\ell_2$-loss function,
we prove that all local minima are global and all saddle points are strict. Thus, our results enable gradient based algorithms with arbitrary initialization to recover the solution for the quadratic feasibility problem.


\section{Properties of the quadratic mapping}\label{sec:injectivity}
As a first step towards establishing our main results, we start by characterizing when the quadratic mapping $\mathcal{M}_{\mathcal{A}}$ defined in~\eqref{eq:defmapping} is in fact an injective mapping. Notice that the injectivity of the mapping is equivalent to the problem being identifiable (and hence solvable). It is also worth noting that in the context of the phase retrieval problem, when $\mathcal{M}_\mathcal{A}$ is injective, it is said to possess the phase retrievability property \cite{wang2017generalized}.

\begin{lemma}\label{lem:existence}
The following statements are equivalent:
\begin{enumerate}
    \item \label{lem:injectivity1}The nonlinear map $\mathcal{M}_{\mathcal{A}} : \widehat{\mathbb{C}}^n \rightarrow \mathbb{C}^m$ is injective.
    \item \label{lem:bilipschitz}
    There exist constants $\alpha ,\beta >0$ such that $\forall \mathbf{u}, \mathbf{v}\in \mathbb{C}^n$, 
\begin{equation}\label{eq:bilipschitz}
    \beta \|[[\mathbf{u},\mathbf{v}]]\|^2_F 
    \geq
    \sum_{i=1}^m|\langle A_{i}, [[\mathbf{u},\mathbf{v}]] 
    \rangle|^2
    \geq \alpha \|[[\mathbf{u},\mathbf{v}]]\|^2_F. 
\end{equation}
\end{enumerate}
\end{lemma}
We refer the reader to Appendix \ref{lem:existence_proof} for a detailed proof. Lemma~\ref{lem:existence} characterizes the injectivity 
{property} of $\mathcal{M}_{\mathcal{A}}$ in terms of the action of the $A_i$ matrices on the outer product $[[\mathbf{u},\mathbf{v}]]$. In particular, it
says that the mapping $\mathcal{M}_{\mathcal{A}}$ is injective if and only if the matrices $A_i$ do not distort the outer-product $[[\mathbf{u}, \mathbf{v}]]$ too much. We use this characterization to obtain a more tractable condition that allows us to establish the injectivity of $\mathcal{M}_{\mathcal{A}}$ in the case of Gaussian measurement matrices ({forthcoming in}
Lemma~\ref{thm:bilip}).


 

Our tractable condition is based on what we call the \emph{stability} of the mapping. 
{Informally speaking, the mapping} $\mathcal{M}_{\mathcal{A}}$ is stable 
{if for any given $\epsilon>0$, there exists $\delta>0$ such that 
$\|\mathcal{M}_\mathcal{A}(\mathbf{x}) - \mathcal{M}_\mathcal{A}(\mathbf{y})\|\geq \delta$
whenever the vectors $\mathbf{x}, \mathbf{y}\in \widehat{\mathbb{C}}^n$ satisfy $d(\mathbf{x}, \mathbf{y})\geq \epsilon$.}
 Such stability properties have been considered in the specific context of phase retrieval; see e.g., \cite{eldar2014phase, candes2013phaselift, duchi2017solving, ye2003new}. 

\begin{definition}[($\alpha, \beta$)-stability]\label{def:stability}
{Consider the mapping $\mathcal{M}_\mathcal{A}$ defined in~\eqref{eq:defmapping}. We say that $\mathcal{M}_\mathcal{A}$ is ($\alpha, \beta$)-stable
in a metric $d$ on $\mathbb{C}^n$, with
$0 < \alpha \leq \beta$, if the following relation holds}
for all $\mathbf{x}_1, \mathbf{x}_2\in \mathbb{C}^n$:
\begin{equation}
    \alpha d(\mathbf{x}_1, \mathbf{x}_2) \leq \|\mathcal{M}_\mathcal{A}(\mathbf{x}_1) - \mathcal{M}_\mathcal{A}(\mathbf{x}_2)\|_2 \leq \beta d(\mathbf{x}_1, \mathbf{x}_2).
\end{equation}
\end{definition}
{The constants $\alpha, \beta$ depend on the choice of the metric $d(\cdot,\cdot)$. 
Throughout the paper, we will
work with the metric $d$ as defined in~\eqref{eq:defdist}.} 
Note that our definition of $(\alpha, \beta)$-stability sandwiches the norm of the measurement differences, which is in contrast to the stability concept considered for phase retrieval in \cite{eldar2014phase, duchi2017solving}. The constants $\alpha, \beta$ can also be thought of as a {\it condition number}, thereby allowing one to {quantify} the quality of the map;  {the higher the ratio between $\alpha$ and $\beta$, 
the better the ability of the mapping to  
distinguish between two distinct inputs.} 


{With this definition of stability in place, 
the next lemma states that the map
$\mathcal{M}_\mathcal{A}$ is injective if and only if 
it is stable.}
\begin{lemma}\label{thm:bilip}
The mapping $\mathcal{M}_{\mathcal{A}}$ is injective iff it is $(\alpha,\beta)$-stable for some constants $0 <\alpha \leq \beta$.
\end{lemma}
Please refer to Appendix \ref{sec:pf_bilip} for a proof of the lemma. This result demonstrates the usefulness of both our choice of the metric $d(\cdot, \cdot)$ from \eqref{eq:defdist} and of our definition of stability (Definition~\ref{def:stability}). 
{As we will see in what follows,} 
Lemma~\ref{thm:bilip} allows one to assess the conditions under which the measurement model implied by the mapping $\mathcal{M}_\mathcal{A}$ is \emph{identifiable}. 




Next, we turn our attention to the question of when one can establish the above stability condition, and thereby establish the identifiability of the underlying measurement model. 
{To do so}, 
we take our cues from the compressive sensing and phase retrieval literature, and we assume that 
{Hermitian matrices in the set 
$\mathcal{A}= \{A_i\}_{i=1}^m$} are sampled from 
a complex Gaussian distribution. More precisely, 
we {\it assume} that each entry in the upper triangle (including the diagonal) of $A_i$ is drawn independently from $\mathcal{N}(0,1)$. The remaining entries are determined by the fact that the  {matrices $A_i$ are Hermitian.}

{We first observe that, 
when $\mathcal{A}$ is chosen as described above, then 
$\mathbb{E}\left[ \left|\langle A_i, \mathbf{x} \mathbf{x}^\ast - \mathbf{y} \mathbf{y}^\ast\rangle\right|^2\right] = d(\mathbf{x}, \mathbf{y})^2$ for all $\mathbf{x}, \mathbf{y}\in \mathbb{C}^n$, and for all $i\in [m]$} (see Appendix~\ref{lem:subexpo_fro_proof} for more details). 
Our next result shows that these quantities are actually concentrated around their expected values.  

\begin{lemma}\label{lem:subexpo_fro}
Let $\mathcal{A} = \{A_i\}_{i=1}^m$ be a set of complex {Hermitian Gaussian random matrices} for the measurement model given by \eqref{eq:main_prob}. Then, given $\epsilon>0$ {and vectors $\mathbf{x}, \mathbf{y}
\in\mathbb{C}^n$}, there are constants $c,d>0$ such that 
\begin{align}
        &\ \mathbb{P}\left(\left|\sum_{i=1}^ {m} \frac{1}{m}|\langle A_i, \mathbf{x}\mathbf{x}^* - \mathbf{y}\mathbf{y}^*\rangle|^2 - d(\mathbf{x}, \mathbf{y})^2 \right| \geq \epsilon d(\mathbf{x}, \mathbf{y})^2 \right)\nonumber\\ & \quad 
        \leq de^{-cm\epsilon}.
\end{align}
\end{lemma}
Please see Appendix~\ref{lem:subexpo_fro_proof} for a proof of {Lemma~\ref{lem:subexpo_fro}.} We would like to draw  attention to the fact that proving such a concentration result requires careful attention to the fact that we are operating in the complex domain; our proof techniques may be of independent interest. 
Notice that this concentration result only holds for a fixed pair of vectors. We however need 
 {the concentration} result to be uniform over all pairs of vectors in $\mathbb{C}^n$. For this, we will next adapt a standard covering argument as follows; we use $S^{n-1}$ to denote the $n-$dimensional unit sphere in $\mathbb{C}^n$.

\begin{lemma}\label{lem:covering}
 {Given $\delta>0$, let $\mathcal{N}_\delta$
be the smallest collection of $n$-dimensional 
balls of radius $\delta$ whose union covers the sphere $S^{n-1}$. 
Then, for any matrix $A\in \mathbb{C}^{n\times n}$, 
we have}
\begin{align}
    &\ (1-2\delta)\sup_{\mathbf{x}_1, \mathbf{x}_2\in S^{n-1}}|\langle A, \mathbf{x}_1\mathbf{x}_1^* - \mathbf{x}_2\mathbf{x}_2^*\rangle| \leq \sup_{\mathbf{x}_1, \mathbf{x}_2\in \mathcal{N}_\delta}|\langle A, \mathbf{x}_1\mathbf{x}_1^* - \mathbf{x}_2\mathbf{x}_2^*\rangle| \leq (1+2\delta)\sup_{\mathbf{x}_1, \mathbf{x}_2\in S^{n-1}}|\langle A, \mathbf{x}_1\mathbf{x}_1^* - \mathbf{x}_2\mathbf{x}_2^*\rangle|.\nn
\end{align}
\end{lemma}
We refer the reader to Appendix~\ref{lem:covering_proof} for a proof of this lemma. This argument is not new and has found application in several results; see e.g., \cite{huang2019solving, vershynin2010introduction, meckes2004concentration}.

We are finally ready to state and prove our first main result. 
\begin{theorem}\label{thm:isometry}
Let  {$\mathcal{A} = \{A_i\}_{i=1}^m$ be the set of complex Gaussian random matrices, and assume 
the number of measurements satisfies $m > Cn$.
Then, for any given $\xi\in(0,1)$,}
there exist constants $C, c_0, d_0 > 0$ and 
$\beta \ge \alpha > 0$ 
such that, with probability at least $1 - \xi$, the following relation holds
\[
    \alpha d(\mathbf{x}, \mathbf{y})\leq\|\mathcal{M}_\mathcal{A}(\mathbf{x}) - \mathcal{M}_\mathcal{A}(\mathbf{y})\|_2 \leq \beta d(\mathbf{x}, \mathbf{y}).\nn
\]
\end{theorem}
\begin{proof}
Consider $\mathbf{x}, \mathbf{y}\in \mathbb{C}^n$. {If $\mathbf{x}\sim\mathbf{y}$, then $d(\mathbf{x}, \mathbf{y}) =0 $ and $\mathcal{M}_\mathcal{A}(\mathbf{x}) = \mathcal{M}_\mathcal{A}(\mathbf{y})$, and the result holds trivially.
Therefore, in the sequel, we assume 
that the vectors are distinct, implying that}
$d(\mathbf{x}, \mathbf{y})> 0$. 

From Lemma \ref{lem:subexpo_fro}, for a 
 {given $\epsilon>0$ we have} 
\[
        \mathbb{P}\left(\left|\sum_{i=1}^{ {m}} \frac{1}{m}\frac{|\langle A_i, \mathbf{x}\mathbf{x}^* - \mathbf{y}\mathbf{y}^*\rangle|^2}{d(\mathbf{x}, \mathbf{y})^2} - 1\right| \geq \epsilon \right)
        \leq de^{-cm\epsilon}.
\]
 {According to~\cite{vershynin2010introduction}, 
for any $\delta>0$, we have} the following upper bound on the size of the covering: 
\(|\mathcal{N}_\delta| \leq \left(\frac{12}{\delta}\right)^n.\)
Therefore, for a given $\epsilon,\delta >0$, 
by Lemma~\ref{lem:subexpo_fro} and the preceding 
union bound, we have that 
\begin{align*}
& \ \ \mathbb{P}\left(\sup_{\mathbf{x}, \mathbf{y}\in \mathcal{N}_\delta}\left|\sum_{i=1}^{ {m}} \frac{1}{m}\frac{|\langle A_i, \mathbf{x}\mathbf{x}^* - \mathbf{y}\mathbf{y}^*\rangle|^2}{d(\mathbf{x}, \mathbf{y})^2} - 1 \right| > \epsilon \right) 
\leq d e^{-cm\epsilon}\left(\frac{12}{\delta}\right)^{n},
\end{align*}
where $d, c$ are the same constants as in Lemma~\ref{lem:subexpo_fro}.
This implies that
\begin{align}
    \mathbb{P}\left(\sup_{\mathbf{x}, \mathbf{y}\in \mathcal{N}_\delta}\left|\sum_{i=1}^m \frac{1}{m}\frac{|\langle A_i, \mathbf{x}\mathbf{x}^* - \mathbf{y}\mathbf{y}^*\rangle|^2}{d(\mathbf{x}, \mathbf{y})^2} - 1 \right| \leq \epsilon \right)\geq 1-de^{-cm\epsilon}\left(\frac{12}{\delta}\right)^{n}.\nn
\end{align}
Now, observe that 
\begin{align}
    \sup_{\mathbf{x}, \mathbf{y}\in \mathcal{N}_\delta}\left|\sum_{i=1}^m \frac{1}{m}\frac{|\langle A_i, \mathbf{x}\mathbf{x}^* - \mathbf{y}\mathbf{y}^*\rangle|^2}{d(\mathbf{x}, \mathbf{y})^2} - 1 \right| \geq \sup_{\mathbf{x}, \mathbf{y}\in \mathcal{N}_\delta}
    \sum_{i=1}^m \frac{1}{m}\frac{|\langle A_i, \mathbf{x}\mathbf{x}^* - \mathbf{y}\mathbf{y}^*\rangle|^2}{d(\mathbf{x}, \mathbf{y})^2} - 1.\nonumber
\end{align}
Therefore,
\begin{align}
    \mathbb{P}\left(\sup_{\mathbf{x}, \mathbf{y}\in \mathcal{N}_\delta}
    \sum_{i=1}^m
    \frac{1}{m}\frac{|\langle A_i, \mathbf{x}\mathbf{x}^* - \mathbf{y}\mathbf{y}^*\rangle|^2}{d(\mathbf{x}, \mathbf{y})^2} - 1 \leq \epsilon \right) \geq 1-d e^{-cm\epsilon}
    \left(\frac{12}{\delta}\right)^{n}.\nn
\end{align}
 {By applying the covering result of 
Lemma \ref{lem:covering} to each matrix $A_i$, 
averaging the resulting relations over $m$, and
using}
\begin{equation*}
    \sup_{\mathbf{x}, \mathbf{y}\in \mathcal{N}_\delta}d(\mathbf{x}, \mathbf{y}) \leq (1+2\delta)\sup_{\mathbf{x}, \mathbf{y}\in S^{n-1}}d(\mathbf{x}, \mathbf{y}).
\end{equation*}
we obtain
\begin{align*}
    \sup_{\mathbf{x}, \mathbf{y}\in \mathcal{N}_\delta}\sum_{i=1}^m
    \frac{1}{m}\frac{|\langle A_i, \mathbf{x}\mathbf{x}^* - \mathbf{y}\mathbf{y}^*\rangle|^2}{d(\mathbf{x}, \mathbf{y})^2} \geq \sup_{\mathbf{x}, \mathbf{y}\in S^{n-1}}\sum_{i=1}^m
    \frac{(1-2\delta)^2}{(1+2\delta)^2m}\frac{|\langle A_i, \mathbf{x}\mathbf{x}^* - \mathbf{y}\mathbf{y}^*\rangle|^2}{d(\mathbf{x}, \mathbf{y})^2}.
\end{align*}
Thus, we can conclude that
\begin{align}
    \mathbb{P}\left(\sup_{\mathbf{x}, \mathbf{y}\in S^{n-1}}\sum_{i=1}^m
    \frac{1}{m}\frac{|\langle A_i, \mathbf{x}\mathbf{x}^* - \mathbf{y}\mathbf{y}^*\rangle|^2}{d(\mathbf{x}, \mathbf{y})^2} \leq
    \frac{\left(1+2\delta)^2(1 +\epsilon\right)}{(1-2\delta)^2}\right)\geq 1-de^{-cm\epsilon}\left(\frac{12}{\delta}\right)^{n}.\nn
\end{align}
Similarly, we can prove that
\begin{align}
    \mathbb{P}\left(\inf_{\mathbf{x}, \mathbf{y}\in S^{n-1}}\sum_{i=1}^m
    \frac{1}{m}\frac{|\langle A_i, \mathbf{x}\mathbf{x}^* - \mathbf{y}\mathbf{y}^*\rangle|^2}{d(\mathbf{x}, \mathbf{y})^2} \leq
    \frac{\left(1-2\delta)^2(1 -\epsilon\right)}{(1+2\delta)^2}\right)\geq 1-de^{-cm\epsilon}\left(\frac{12}{\delta}\right)^{n}.\nn
\end{align}
 {Letting $C$ be such that 
$C > \frac{\log 12d - \log \delta\xi}{c\epsilon}$
and letting  $m\geq Cn$, we can see that 
the following relation holds
with probability at least $1-\xi$: for all 
$\mathbf{x},\mathbf{y}\in\mathbb{C}^n$},
\begin{align}
    \beta d(\mathbf{x}, \mathbf{y})\geq\|\mathcal{M}_\mathcal{A}(\mathbf{x}) - \mathcal{M}_\mathcal{A}(\mathbf{y})\|_2 \geq \alpha d(\mathbf{x}, \mathbf{y}),\nn
\end{align}
where $\alpha , \beta $ are given by
\begin{align}
    \alpha \triangleq  \frac{\left( (1-2\delta)^2(1 -\epsilon\right)}{(1+2\delta)^2}, ~~~
    \beta  \triangleq   \frac{\left((1+2\delta)^2(1 +\epsilon \right)}{(1-2\delta)^2}.\nn
\end{align}
Notice that we essentially have a choice of the values of $\delta$ and $\epsilon$. 
The closer they are to 0, the stronger the stability result. However, this also implies the larger $m$, the number of measurements, needs to be. 
\end{proof}

\section{Non-convex loss reformulation and the Optimization Landscape}\label{sec:Opt_landscape}
In the  {preceding}
section, we established conditions under which the mapping $\mathcal{M}_\mathcal{A}$ represents an identifiable measurement model.  {We next consider determining 
a feasible solution by applying efficient methods to the nonconvex optimization-based reformulation of the feasibility problem,
as given in~(P2).}
 {As noted earlier, in general, 
gradient-based aproaches need not converge to a global minimum of the $\ell_2$-loss minimization problem (P2). Lately}, nonconvex optimization 
has received considerable attention. In particular, methods like SGD and other gradient based methods have been shows to be astonishingly successful in converging to global  {minima} in many nonconvex problems \cite{du2018gradient, chen2018gradient, jin2017escape}. Arguably, the reason for this is that the optimization landscape of these (somewhat well-behaved) nonconvex problems enjoy some advantageous properties which are crucial in the empirical success of the gradient based methods  \cite{bower2018landscape, davis2017nonsmooth, sun2018geometric}. The work in \cite{sun2018geometric} proves that the $\ell_2$-loss function for the phase retrieval problem enjoys properties such as all local minima being global, and each 
saddle point having a strictly negative curvature. 
This work adds to this rich body of work by demonstrating similarly advantageous properties of the optimization landscape for the quadratic feasibility problem. Before we proceed 
we observe that, since the $\ell_2$-loss function is not differentiable in the complex space $\mathbb{C}^n$, 
it is challenging to address 
the problem~\eqref{eq:ell2loss0} in a standard way.
In what follows, we instead use techniques from Wirtinger calculus \cite{candes2015phase}. Our first step is to define the notion of a strict saddle function.

\begin{definition}\label{def:strictsaddle}
Let $\beta, \gamma, \zeta$ be positive scalars. 
 {A function $f$ is said be $(\beta, \zeta, \gamma)$-strict saddle function} if for any $\mathbf{x}\in \mathbb{C}^n$, at least one of the following statements 
is true:
\begin{enumerate}
    \item $\|\nabla f(\mathbf{x})\| \geq \beta$;
    \item $\langle \nabla^2f(\mathbf{x})\mathbf{z}, \mathbf{z} \rangle \leq -\zeta$ for some $\mathbf{z} \in \mathbb{C}^n$;
    \item $\mathbf{x}$ is $\gamma$-close to a local minimum, i.e., $d(\mathbf{x},\mathbf{w})\le \gamma$
    for some $\mathbf{w}$ satisfying $\nabla f(\mathbf{w}) = 0$ and $\nabla^2 f(\mathbf{w}) \succeq 0$. 
\end{enumerate}
\end{definition}
Intuitively, this implies that every $\mathbf{x}\in \mathbb{C}^n$ either violates optimality (condition 1 and 2) or is close to a local optimum. A line of recent work \cite{lee2016gradient, lee2017first, jain2017non} has explored the efficacy of gradient based methods in finding a local optimum of functions satisfying Definition~\ref{def:strictsaddle}.

We analyze the optimization landscape of (P2) when our measurement matrices are Hermitian and complex Gaussian, and show that with a high probability \emph{every local minimum is in fact global} (upto the equivalence relation $\sim$). Our next main result states that the function $f$ in (P2) is strict saddle.  
\begin{theorem}\label{thm:strictsaddle}
 {Let $\{A_i\}_{i=1}^m$ be a set of complex $n\times n$ Gaussian random matrices, and let $m > Cn$ for some 
constant $C>0$.
Let the scalars $\{c_i\}_{i=1}^m$ characterizing the objective function $f$ of problem (P2) be generated by quadratic measurements of an unknown vector $\mathbf{z}$.
Then, for any given $\xi\in(0,1)$,} there exist positive constants $\beta, \gamma,$ and $\zeta$ such that 
the following statements hold with probability at least $1-\xi$:
\begin{itemize}
    \item [1)] The function f is $(\beta, \zeta, \gamma)$-strict saddle, and
    \item [2)] Every local minimum $\mathbf{w}$ 
    of $f$ satisfies $d(\mathbf{w}, \mathbf{z}) = 0$
\end{itemize}
\end{theorem}
Recall that the distance metric $d$ is defined on the quotient space $\widehat{\mathbb{C}}^n$, therefore, the second statement above says that \\$\mathbf{w}\sim \mathbf{z}$.  We refer the reader to Appendix~\ref{thm:strictsaddle_proof} for the details, but we will give here the brief idea of the proof. 
\begin{proof}[Proof sketch]
Notice that to show that the function $f$ in (P2) is a strict saddle function, it suffices to only consider the points $\mathbf{x}\in \mathbb{C}^n$ such that $\left\| \nabla f(\mathbf{x})\right\| < \beta$ (otherwise, condition 1 of Definition~\ref{def:strictsaddle} is satisfied). For all such points, we analyze the behavior of the Hessian and establish that there exists a direction $\Delta\in \mathbb{C}^n$ such that the following inequality holds 
\begin{align}\label{eq:sketch_arg}
\langle\nabla^2 f(\mathbf{x})\Delta, \Delta\rangle \leq -c_0\|\mathbf{x}\mathbf{x}^* - \mathbf{z}\mathbf{z}^*\|_F^2,
\end{align}
where $c_0 > 0$ and $\left\| \cdot \right\|_F$ is the standard Frobenius norm. By the equivalence of finite dimensional norms, the term on the right side is positive if and only if $d(\mathbf{x}, \mathbf{z}) > 0$. This of course implies that whenever $d(\mathbf{x}, \mathbf{z}) >0$, there is a direction where the Hessian has a strict negative curvature, and hence such a point cannot be an optimum. In other words, we can conclude that: (1) {All local minima satisfy} $d(\mathbf{x}, \mathbf{z}) = 0$, and (2) all saddle points have a  {strictly negative} curvature. 

This concludes the proof. 
\end{proof}
Finally, we remark that the properties of the optimization landscape that we have established allows one to use any gradient based iterative method to find a global optimum of problem (P2) --  hence, find a solution to the quadratic feasibility problem. 
 {Furthermore, our results above also imply that a
gradient method, with an arbitrary initial point, would work, which is in a sharp contrast with the  existing works,such as~\cite{huang2019solving}.}
Formally, we have the following result. 
\begin{corollary}
Consider a gradient method applied to minimize the function $f$ in~\eqref{eq:ell2loss0}.
Then, for an arbitrary initial point, the method point converges to a global minimum of the loss function $f$ 
associated with the quadratic feasibility problem.
\end{corollary}
 Given the landscape properties we have derived in Theorem~\ref{thm:strictsaddle}, this result follows in a straightforward manner, for instance, from \balance Theorem 4.1 in \cite{lee2016gradient}. We would like to remark here that the broad flow of ideas in our proof of Theorem~\ref{thm:strictsaddle} bears similarities to those in papers like \cite{ge2017no, bhojanapalli2016global}. However, to the best of our knowledge, the present paper is the first to derive such results in the complex domain.
 

\bibliography{mybib}
\bibliographystyle{IEEEtran}
\newpage
\section{Appendix A : Injectivity}
\setcounter{theorem}{1}
\setcounter{lemma}{0}
\subsection{Proof of Lemma \ref{lem:existence}}\label{lem:existence_proof}
\begin{lemma}
The following statements are equivalent:
\begin{enumerate}
    \item \label{lem:injectivity1}The nonlinear map $\mathcal{M}_{\mathcal{A}} : \widehat{\mathbb{C}}^n \rightarrow \mathbb{C}^m$ is injective.
    \item
    There exist constants $\alpha ,\beta >0$ such that $\forall \mathbf{u}, \mathbf{v}\in \mathbb{C}^n$, 
\begin{equation}
    \beta \|[[\mathbf{u},\mathbf{v}]]\|^2_F 
    \geq
    \sum_{i=1}^m|\langle A_{i}, [[\mathbf{u},\mathbf{v}]] 
    \rangle|^2
    \geq \alpha \|[[\mathbf{u},\mathbf{v}]]\|^2_F. \nonumber\\
\end{equation}
\end{enumerate}
\end{lemma}
\begin{proof}
( \ref{lem:injectivity1}$\Rightarrow$\ref{lem:bilipschitz} )

The following result from \cite{wang2017generalized} is quite crucial,
\begin{theorem*}[Theorem 2.1, \cite{wang2017generalized}]\label{thm:bilinear_r}
Let $\mathcal{A} = \{A_i\}_{i=1}^m\subset \mathbf{H}_n(\mathbb{C})$. The following statements are equivalent:
\begin{enumerate}
    \item For a given $\mathcal{A} = \{A_i\}_{i=1}^m$, the mapping $\mathcal{M}_\mathcal{A}$ has phase retrieval property.
    \item There exists no nonzero vector $\mathbf{v}, \mathbf{u}\in \mathbb{C}^n$ with $\mathbf{u} \neq ic\mathbf{v}$, $c\in \mathbb{R}$, such that Re($\langle A_j\mathbf{u}, \mathbf{v}\rangle) = 0$ for all $1\leq j\leq m$.
\end{enumerate}
\end{theorem*}

For the mapping $\mathcal{M}_\mathcal{A}$ to be injective, the following should holds,
\begin{equation}
    \mathcal{M}_\mathcal{A}(\mathbf{x}) = \mathcal{M}_\mathcal{A}(\mathbf{y})~~~\text{iff}~~~\mathbf{x}\sim\mathbf{y}\nn
\end{equation}
Hence for $\mathbf{x} \nsim  \mathbf{y}$, $\mathcal{M}_\mathcal{A}(\mathbf{x})\neq \mathcal{M}_\mathcal{A}(\mathbf{y})$.
It can be verified that for any $\phi\in [0,2\pi]$, $\textbf{u} =\textbf{x}-e^{i\phi}\textbf{y}, \textbf{v} =\textbf{x}+e^{i\phi}\textbf{y}$ satisfies the following transformation, 
\begin{equation}\label{eq:switchu2x}
    \left(\mathbf{x}\mathbf{x}^* - \mathbf{y}\mathbf{y}^*\right) = \left(\mathbf{u}\mathbf{v}^* + \mathbf{v}\mathbf{u}^*\right) = [[\mathbf{u}, \mathbf{v}]]
\end{equation}
Thus,
\begin{align}
    \|\mathcal{M}_\mathcal{A} (\mathbf{x}) - \mathcal{M}_\mathcal{A} (\mathbf{y})\|_2^2= \sum_{i=1}^m\left|\langle A_i, [[\mathbf{u}, \mathbf{v}]]\rangle \right|^2\nonumber
\end{align}
We define the lower bound $\alpha $ and upper bound $\beta $ as below,
\begin{align}
    \alpha  := \min_{T\in S^{1,1}, \|T\|_F = 1}\sum_{i=1}^m |\langle A_i, T\rangle|^2, ~~~
    \beta  := \max_{T\in S^{1,1}, \|T\|_F = 1}\sum_{i=1}^m |\langle A_i, T\rangle|^2\nonumber
\end{align}

The set $T\in S^{1,1}, \|T\|_F = 1$ is compact, hence the constants $\alpha , \beta $ exists.

From Theorem~\ref{thm:bilinear_r} statement, it is clear that $\langle A_i, [[\mathbf{u}, \mathbf{v}]]\rangle = \langle A_i, \mathbf{u} \mathbf{v}^*\rangle + \langle A_i, \mathbf{v} \mathbf{u}^*\rangle = Re(\langle A_i\mathbf{v}, \mathbf{u}\rangle) = 0$ cannot be satisfied for all $i\in [1,m]$ if $\mathbf{x}\nsim\mathbf{y}$ satisfying equation \eqref{eq:switchu2x}.

( \ref{lem:injectivity1}$\Leftarrow$\ref{lem:bilipschitz} )

Instead of proving \ref{lem:injectivity1}$\Leftarrow$\ref{lem:bilipschitz} , we argue the negation holds, i.e. \ref{lem:injectivity1}$\nRightarrow$\ref{lem:bilipschitz}. 

Suppose the mapping $\mathcal{M}_\mathcal{A}$ is not injective.

Then $\exists \mathbf{x}, \mathbf{y}\in \mathbb{C}^n$ such that,
\begin{equation}
    \mathbf{x}\nsim \mathbf{y}, ~~~\mathcal{M}_\mathcal{A}(\mathbf{y})=\mathcal{M}_\mathcal{A}(\mathbf{x})\nonumber
\end{equation}
Thus $\|\mathbf{x}\mathbf{x}^* - \mathbf{y}\mathbf{y}^*\|_F\neq 0$, but $\|\mathcal{M}_\mathcal{A}(\mathbf{y})-\mathcal{M}_\mathcal{A}(\mathbf{x})\|_2 = 0$. Thus $\alpha = 0$ and hence the negation follows.
\end{proof}

\subsection{Proof of Lemma \ref{thm:bilip}}\label{sec:pf_bilip}
\begin{lemma}
The mapping $\mathcal{M}_{\mathcal{A}}$ is injective iff it is $(\alpha,\beta)$-stable for some constants $0 <\alpha \leq \beta$.
\end{lemma}
\begin{proof}In order to prove the theorem statement, we examine the properties of the ratio,
\begin{align}
    V(\mathbf{x}, \mathbf{y}) = \frac{\sum_{i=1}^m |\langle A_i, \mathbf{x}\mathbf{x}^* - \mathbf{y}\mathbf{y}^*\rangle|^2}{\|\mathbf{x}\mathbf{x}^* - \mathbf{y}\mathbf{y}^*\|_F^2}\nn
\end{align}
The ($\alpha, \beta$)-stability of the mapping $\mathcal{M}_\mathcal{A}$ directly follows from Lemma \ref{lem:existence} and the existence of $\mathbf{u}, \mathbf{v}\in \mathbb{C}^n$ satisfying equation~\eqref{eq:switchu2x}.
\end{proof}

\section{Appendix B : High probability bounds}
\subsection{Proof of Lemma \ref{lem:subexpo_fro}}\label{lem:subexpo_fro_proof}
\begin{lemma}
Let $\mathcal{A} = \{A_i\}_{i=1}^m$ be a set of complex {Hermitian Gaussian random matrices} for the measurement model given by \eqref{eq:main_prob}. Then, given $\epsilon>0$ {and vectors $\mathbf{x}, \mathbf{y}
\in\mathbb{C}^n$}, there are constants $c,d>0$ such that 
\begin{align}
        \mathbb{P}\left(\left|\sum_{i=1}^ {m} \frac{1}{m}|\langle A_i, \mathbf{x}\mathbf{x}^* - \mathbf{y}\mathbf{y}^*\rangle|^2 - d(\mathbf{x}, \mathbf{y})^2 \right| \geq \epsilon d(\mathbf{x}, \mathbf{y})^2 \right) 
        \leq de^{-cm\epsilon}.\nonumber
\end{align}
\end{lemma}
\begin{proof}
A matrix $A\in \mathbb{C}^{n\times n}$ is a complex Hermitian Gaussian random matrix, if,
\begin{enumerate}
    \item $\forall~i$, $a_{ii}\sim \mathcal{N}(0, \sigma^2)$.
    \item $\forall~i, j$, $i\neq j$, $a_{ij}\sim \mathcal{N}(0,\frac{\sigma^2}{2}) + i\mathcal{N}(0,\frac{\sigma^2}{2})$.
\end{enumerate}
Let $\{A_d\}_{d=1}^m$ be a set of complex Hermitian Gaussian random matrices. Define the random variable $Y$,

\begin{align}
    Y&\ = \frac{1}{m} \sum_{d=1}^m|\langle A_d, \mathbf{x}\mathbf{x}^* - \mathbf{y}\mathbf{y}^*\rangle|^2 = \frac{1}{m}\sum_{d=1}^m(\langle A_d, \mathbf{x}\mathbf{x}^* - \mathbf{y}\mathbf{y}^*\rangle)(\langle A_d, \mathbf{x}\mathbf{x}^* - \mathbf{y}\mathbf{y}^*\rangle)\nonumber\\
    &\ = \frac{1}{m} \sum_{d=1}^m\left(\sum_{ij}a_{ij}(x_i\bar{x}_j - y_i\bar{y}_j)\right)\left(\sum_{kl}a_{kl}(x_k\bar{x}_l - y_k\bar{y}_l)\right)\nonumber
\end{align}
Expectation of Y can be evaluated as,
\begin{align}\label{eq:expectationY}
    \mathbb{E}[Y] = \mathbb{E}\left(\frac{1}{m} \sum_{d=1}^m\left(\sum_{ij}a_{ij}(x_i\bar{x}_j - y_i\bar{y}_j)\right)\left(\sum_{kl}a_{kl}(x_k\bar{x}_l - y_k\bar{y}_l)\right)\right)
\end{align}
For every matrix $A_d$, we can split the entire summation \eqref{eq:expectationY} into the following 4 sets:
\begin{enumerate}
    \item $B:= \{(i,j,k,l) | i = j = k =l\}$
    \item $C:= \{(i,j,k,l) | i=l, j=k\}\cap A^C$
    \item $D:= \{(i,j,k,l) | i=k, j=l\}\cap A^C$
    \item $E:= \{(i,j,k,l)\}\cap A^C \cap B^C \cap C^C$
\end{enumerate}
Calculating the expectation of the sum of the elements in each individual sets:
\begin{enumerate}
    \item Set B,
    
    \begin{align}
        \mathbb{E}\left(\sum_{(i,j,k,l)\in B}a_{ij}a_{kl}(x_i\bar{x}_j - y_i\bar{y}_j)(x_k\bar{x}_l - y_k\bar{y}_l)\right) &\ = \mathbb{E}\left( \sum_{i=1}^n|a_{ii}|^2(|x_i|^2 - |y_i|^2)^2\right)\nonumber\\ & = \sigma^2\sum_{i=1}^n|x_i|^4 + |y_i|^4 -2|x_i|^2|y_i|^2\nn
    \end{align}
    \item Set C, note that for every Hermitian matrix $A_d$, $a_{ij}= \bar{a}_{ji}$
    \begin{align}
        & \mathbb{E}\left(\sum_{(i,j,k,l)\in C}a_{ij}a_{kl}(x_i\bar{x}_j - y_i\bar{y}_j)(x_k\bar{x}_l - y_k\bar{y}_l)\right)\nonumber\\ &\ = \mathbb{E}\left( \sum_{i,j=1, i\neq j}^n|a_{ij}|^2(|x_i|^2|x_j|^2 -y_i\bar{y}_jx_j\bar{x}_i -x_i\bar{x}_jy_j\bar{y}_i + |y_i|^2|y_j|^2)\right)\nonumber\\ &\ = \sigma^2\sum_{i, j=1, i\neq j}^n(|x_i|^2|x_j|^2 -y_i\bar{y}_jx_j\bar{x}_i - x_i\bar{x}_jy_j\bar{y}_i + |y_i|^2|y_j|^2)\nonumber
    \end{align}
    \item Set D,
    \begin{align}
        & \mathbb{E}\left(\sum_{(i,j,k,l)\in D}a_{ij}a_{kl}(x_i\bar{x}_j - y_i\bar{y}_j)(x_k\bar{x}_l - y_k\bar{y}_l)\right)= \mathbb{E}\left( \sum_{ij}(a_{ij})^2(x_i\bar{x}_j - y_i\bar{y}_j)^2)\right) = 0\nonumber
    \end{align}
    Notice that $\forall i,j$
\begin{align}
    (a_{ij})^2&\ = ((a_{ij}^r)^2 - (a_{ij}^i)^2 + ia_{ij}^ra_{ij}^i)\nonumber
\end{align}
Thus,
\begin{align}
\mathbb{E}\left[(a_{ij})^2\right]= \mathbb{E}\left[(a_{ij})^2\right]= \mathbb{E}\left[((a_{ij}^r)^2 - (a_{ij}^i)^2 + ia_{ij}^ra_{ij}^i)\right]\nonumber
\end{align}
Since both the real and imaginary parts are independent, we can conclude, $\mathbb{E}\left[((a_{ij}^r)^2 - (a_{ij}^i)^2 + ia_{ij}^ra_{ij}^i)\right] = 0$
\item Set E, 

All elements $a_{ij}, a_{kl}$ are independent in $(i,j,k,l)\in E$,
    \begin{align}
        \mathbb{E}\left(\sum_{(i,j,k,l)\in E}a_{ij}a_{kl}(x_i\bar{x}_j - y_i\bar{y}_j)(x_k\bar{x}_l - y_k\bar{y}_l)\right) = 0\nonumber
    \end{align}
\end{enumerate}
In conclusion,
\begin{align}
    \mathbb{E}[Y] &\ = \sigma^2 \left(\left(\sum_{i=1}^n |x_i|^2\right)^2 + \left(\sum_{i=1}^n |y_i|^2\right)^2 - 2\sum_{i=1}^n|x_i|^2|y_i|^2 - \sum_{i, j, i\neq j}y_i\bar{y}_jx_j\bar{x}_i - \sum_{i, j, i\neq j}y_j\bar{y}_ix_i\bar{x}_j\right)\nonumber\\
    &\ = \sigma^2\left[\|\mathbf{x}\|_2^4 + \|\mathbf{y}\|_2^4 - |\langle \mathbf{x}, \mathbf{y}\rangle|^2\right]\nonumber
\end{align}
From Lemma 3.9~\cite{balan2016reconstruction}, note that $tr\{(\mathbf{x}\mathbf{x}^* - \mathbf{y}\mathbf{y}^*)^2\} = \left[\|\mathbf{x}\|_2^4 + \|\mathbf{y}\|_2^4 - |\langle \mathbf{x}, \mathbf{y}\rangle|^2\right]$, where $tr\{\cdot\}$ represents the trace of a matrix. Since $\mathbf{x}\mathbf{x}^* - \mathbf{y}\mathbf{y}^*$ is a Hermitian matrix $tr\{(\mathbf{x}\mathbf{x}^* - \mathbf{y}\mathbf{y}^*)^2\} = \|\mathbf{x}\mathbf{x}^* - \mathbf{y}\mathbf{y}^*\|_F^2$. Hence, finally we can state,
\begin{equation}
    \mathbb{E}[Y] = \|\mathbf{x}\mathbf{x}^* - \mathbf{y}\mathbf{y}^*\|_F^2\nonumber
\end{equation}
Next we focus on obtaining concentration bounds. Just as with expectation $\mathbb{E}[Y]$, we evaluate the behaviour of deviation in each individual set $B,C,D$ and $E$.
\begin{enumerate}
    \item Set B,
    \begin{align}
        & \sum_{(i,j,k,l)\in B}a_{ij}a_{kl}(x_i\bar{x}_j - y_i\bar{y}_j)(x_k\bar{x}_l - y_k\bar{y}_l) -\mathbb{E}\left(\sum_{(i,j,k,l)\in B}a_{ij}a_{kl}(x_i\bar{x}_j - y_i\bar{y}_j)(x_k\bar{x}_l - y_k\bar{y}_l)\right) \nonumber\\
        & = \sum_{i=1}^n\left(|a_{ii}|^2 -\sigma^2\right)\left(|x_i|^4 + |y_i|^4 -2|x_i|^2|y_i|^2\right)\nonumber
    \end{align}
    Note that $\forall i\in [1, n], |a_{ii}|^2 - \sigma^2$ is a centered subexponential random variable. 
    
    \item Set C,
    \begin{align}
        & \sum_{(i,j,k,l)\in C}a_{ij}a_{kl}(x_i\bar{x}_j - y_i\bar{y}_j)(x_k\bar{x}_l - y_k\bar{y}_l) -\mathbb{E}\left(\sum_{(i,j,k,l)\in C}a_{ij}a_{kl}(x_i\bar{x}_j - y_i\bar{y}_j)(x_k\bar{x}_l - y_k\bar{y}_l)\right) \nonumber\\
        & =  \sum_{i,j=1, i\neq j}^n(|a_{ij}|^2-\sigma^2)(|x_i|^2|x_j|^2 -y_i\bar{y}_jx_j\bar{x}_i - x_i\bar{x}_jy_j\bar{y}_i + |y_i|^2|y_j|^2)\nonumber
    \end{align}
    Again, note that $\forall i,j\in [1,n]^2, i\neq j, |a_{ij}|^2 - \sigma^2$ is a centered subexponential random variable. 
    
    \item Set D,
    \begin{align}
        & \sum_{(i,j,k,l)\in D}a_{ij}a_{kl}(x_i\bar{x}_j - y_i\bar{y}_j)(x_k\bar{x}_l - y_k\bar{y}_l) -\mathbb{E}\left(\sum_{(i,j,k,l)\in D}a_{ij}a_{kl}(x_i\bar{x}_j - y_i\bar{y}_j)(x_k\bar{x}_l - y_k\bar{y}_l)\right) \nonumber\\
        & =  \sum_{i,j=1, i\neq j}^{n}(a_{ij})^2(x_i)^2(\bar{x}_j)^2\nonumber
    \end{align}
    Note that $a_{ij}^2 = (a_{ij}^r)^2 - (a_{ij}^i)^2 + ia_{ij}^ra_{ij}^i$. This makes it easier to argue that $\forall i,j\in [1,n]^2, i\neq j, (a_{ij})^2$ is a centered subexponential random variable. 
    
    \item For elements in set E,
    \begin{align}
        & \sum_{(i,j,k,l)\in E}a_{ij}a_{kl}(x_i\bar{x}_j - y_i\bar{y}_j)(x_k\bar{x}_l - y_k\bar{y}_l) -\mathbb{E}\left(\sum_{(i,j,k,l)\in E}a_{ij}a_{kl}(x_i\bar{x}_j - y_i\bar{y}_j)(x_k\bar{x}_l - y_k\bar{y}_l)\right) \nonumber\\
        & =  \sum_{(i,j,k,l)\in E}a_{ij}a_{kl}(x_i\bar{x}_j - y_i\bar{y}_j)(x_k\bar{x}_l - y_k\bar{y}_l)\nonumber
    \end{align}
    Since $a_{ij}, a_{kl}$ for $(i,j,k,l)\in E$ are independent, it can be easily seen that $a_{ij}a_{kl}$ is a centered subexponential random variable $\forall (i,j,k,l) \in E$. 
    
\end{enumerate}
Take $\sigma^2 = 1$. We then have the Bernstein type inequality \cite{vershynin2010introduction} as,
\begin{align}
    &\ \mathbb{P}\left(\left|\sum_{d=1}^n \frac{1}{m}|\langle A_d, \mathbf{x}\mathbf{x}^* - \mathbf{y}\mathbf{y}^*\rangle|^2 - \|\mathbf{x}\mathbf{x}^* - \mathbf{y}\mathbf{y}^*\|_F^2 \right|\leq t \right)\nonumber\\ &\ \geq 1 - c_0\text{exp}\left(-c_1m\min\left\{\frac{t^2}{K_4^2\|\mathbf{x}\mathbf{x}^* - \mathbf{y}\mathbf{y}^*\|_2^4}, \frac{t}{K_4\|\mathbf{x}\mathbf{x}^* - \mathbf{y}\mathbf{y}^*\|_\infty^2}\right\}\right)\nn
\end{align}
for some constants $c_0, c_1 >0$.

We introduce the normalized variable $ \epsilon = \frac{t}{\|\mathbf{x}\mathbf{x}^* - \mathbf{y}\mathbf{y}^*\|_F^2}$,
\begin{align}
    \mathbb{P}\left(\left|\sum_{d=1}^n \frac{1}{m}|\langle A_d, \mathbf{x}\mathbf{x}^* - \mathbf{y}\mathbf{y}^*\rangle|^2 - \|\mathbf{x}\mathbf{x}^* - \mathbf{y}\mathbf{y}^*\|_F^2 \right|\geq \epsilon\|\mathbf{x}\mathbf{x}^* - \mathbf{y}\mathbf{y}^*\|_F^2 \right) \leq c_0\text{exp}^{-c_1mE(\epsilon)}\nonumber
\end{align}
where $E(\epsilon) := \min\left\{\frac{\epsilon^2}{K^2}, \frac{\epsilon}{K}\right\}$.

Note that $\|\mathbf{x}\mathbf{x}^* - \mathbf{y}\mathbf{y}^*\|_F^2$ is the distance metric $d(\cdot, \cdot)$ defined in \eqref{eq:defdist}. Hence we can rewrite the high probability result more consicely as,
\begin{align}
    \mathbb{P}\left(\left|\sum_{d=1}^n \frac{1}{m}|\langle A_d, \mathbf{x}\mathbf{x}^* - \mathbf{y}\mathbf{y}^*\rangle|^2 - d(\mathbf{x}, \mathbf{y})^2 \right|\geq \epsilon d(\mathbf{x}, \mathbf{y})^2 \right) \leq c_0\text{exp}^{-c_1mE(\epsilon)}\nonumber
\end{align}
 
\end{proof}

\subsection{Proof of Lemma \ref{lem:covering}}\label{lem:covering_proof}
\begin{lemma}
 {Given $\delta>0$, let $\mathcal{N}_\delta$
be the smallest collection of $n$-dimensional 
balls of radius $\delta$ whose union covers the sphere $S^{n-1}$. 
Then, for any matrix $A\in \mathbb{C}^{n\times n}$, 
we have}
\begin{align}
    (1-2\delta)\sup_{\mathbf{x}_1, \mathbf{x}_2\in S^{n-1}}|\langle A, \mathbf{x}_1\mathbf{x}_1^* - \mathbf{x}_2\mathbf{x}_2^*\rangle| \leq \sup_{\mathbf{x}_1, \mathbf{x}_2\in \mathcal{N}_\delta}|\langle A, \mathbf{x}_1\mathbf{x}_1^* - \mathbf{x}_2\mathbf{x}_2^*\rangle| \leq (1+2\delta)\sup_{\mathbf{x}_1, \mathbf{x}_2\in S^{n-1}}|\langle A, \mathbf{x}_1\mathbf{x}_1^* - \mathbf{x}_2\mathbf{x}_2^*\rangle|.\nn
\end{align}
\end{lemma}
\begin{proof}
In the proof, we relate the supremum of $|\langle A, \mathbf{x}_1\mathbf{x}_1^* - \mathbf{x}_2\mathbf{x}_2^*\rangle|$ over $\mathbf{x}, \mathbf{y}\in S^{n-1}$ to its supremum over $\mathbf{x}, \mathbf{y}\in \mathcal{N}_{\delta}$.
 
Since $\mathcal{N}_\delta$ covers $S^{n-1}$, $\forall \mathbf{x}\in S^{n-1}$,  $\exists \mathbf{u}\in \mathcal{N}_\delta$ such that $\|\mathbf{x}-\mathbf{u}\| \leq \delta$. 
 
Hence $\forall \mathbf{x}_1, \mathbf{x}_2 \in S^{n-1}$, $\exists \mathbf{y}_1, \mathbf{y}_2 \in \mathcal{N}_\delta$ such that,
\begin{align}
    &\ |\langle A, \mathbf{x}_1\mathbf{x}_1^* - \mathbf{x}_2\mathbf{x}_2^*\rangle - \langle A, \mathbf{y}_1\mathbf{y}_1^* - \mathbf{y}_2\mathbf{y}_2^*\rangle|\nonumber\\ 
    &\ = |\langle A\mathbf{x}_1, \mathbf{x}_1\rangle - \langle A\mathbf{y}_1, \mathbf{y}_1\rangle| - \langle A\mathbf{x}_2, \mathbf{x}_2\rangle - \langle A\mathbf{y}_2, \mathbf{y}_2\rangle|\nonumber\\
    &\ =  |\langle A\mathbf{x}_1, \mathbf{x}_1\rangle - \langle A\mathbf{x}_1, \mathbf{y}_1\rangle + \langle A\mathbf{x}_1, \mathbf{y}_1\rangle - \langle A\mathbf{y}_1, \mathbf{y}_1\rangle| +  |\langle A\mathbf{x}_2, \mathbf{x}_2\rangle - \langle A\mathbf{x}_2, \mathbf{y}_2\rangle + \langle A\mathbf{x}_2, \mathbf{y}_2\rangle - \langle A\mathbf{y}_2, \mathbf{y}_2\rangle| \nonumber\\
    &\ = |\langle A\mathbf{x}_1, \mathbf{x}_1 - \mathbf{y}_1\rangle + \langle A\mathbf{y}_1, \mathbf{x}_1-\mathbf{y}_1\rangle| + |\langle A\mathbf{x}_2, \mathbf{x}_2 - \mathbf{y}_2\rangle + \langle A\mathbf{y}_2, \mathbf{x}_2-\mathbf{y}_2\rangle|\nonumber\\
    &\ \leq 2\|A\|\|\mathbf{x}_1-\mathbf{y}_1\| + 2\|A\|\|\mathbf{x}_2-\mathbf{y}_2\|\nonumber\\
    &\ \leq 4\delta\|A\|\nonumber
\end{align}
where $\|A\|$ denotes the spectral norm of the matrix $A$, i.e., $$\left\| A \right\|= \sup_{\mathbf{x}\in S^{n-1}}|\langle A\mathbf{x}, \mathbf{x}\rangle| = \frac{1}{2} \sup_{\mathbf{x}\in S^{n-1}}|\langle A\mathbf{x}, \mathbf{x}\rangle| + \frac{1}{2}\sup_{\mathbf{y}\in S^{n-1}}|\langle A\mathbf{y}, \mathbf{y}\rangle|=  \frac{1}{2}\sup_{\mathbf{x}, \mathbf{y}\in S^{n-1}}|\langle A, \mathbf{x}\mathbf{x}^* - \mathbf{y}\mathbf{y}^*\rangle|$$

We conclude,
\begin{align}
    &\ |\langle A, \mathbf{x}_1\mathbf{x}_1^* - \mathbf{x}_2\mathbf{x}_2^*\rangle| - |\langle A, \mathbf{y}_1\mathbf{y}_1^* - \mathbf{y}_2\mathbf{y}_2^*\rangle| \leq 4\delta\|A\|\nonumber\\
    &\ |\langle A, \mathbf{y}_1\mathbf{y}_1^* - \mathbf{y}_2\mathbf{y}_2^*\rangle| \geq |\langle A, \mathbf{x}_1\mathbf{x}_1^* - \mathbf{x}_2\mathbf{x}_2^*\rangle| - 4\delta\|A\|\nonumber
\end{align}
And,
\begin{align}
    &\ |\langle A, \mathbf{y}_1\mathbf{y}_1^* - \mathbf{y}_2\mathbf{y}_2^*\rangle| - |\langle A, \mathbf{x}_1\mathbf{x}_1^* - \mathbf{x}_2\mathbf{x}_2^*\rangle| \leq 4\delta\|A\|\nonumber\\
    &\ |\langle A, \mathbf{y}_1\mathbf{y}_1^* - \mathbf{y}_2\mathbf{y}_2^*\rangle| \leq |\langle A, \mathbf{x}_1\mathbf{x}_1^* - \mathbf{x}_2\mathbf{x}_2^*\rangle| + 4\delta\|A\|\nonumber
\end{align}
Taking supremum,
\begin{align}\label{eq:coveringconnection}
    \sup_{\mathbf{x}\in \mathcal{N}_\delta}|\langle A, \mathbf{x}_1\mathbf{x}_1^* - \mathbf{x}_2\mathbf{x}_2^*\rangle| &\ \geq \sup_{\mathbf{x}\in S^{n-1}}|\langle A, \mathbf{x}_1\mathbf{x}_1^* - \mathbf{x}_2\mathbf{x}_2^*\rangle| - 4\delta\|A\|\nonumber\\ &\ = (2-4\delta)\|A\| \nonumber\\ &\ = (1-2\delta)\sup_{\mathbf{x}\in S^{n-1}}|\langle A, \mathbf{x}_1\mathbf{x}_1^* - \mathbf{x}_2\mathbf{x}_2^*\rangle|\nonumber\\
    \sup_{\mathbf{x}\in \mathcal{N}_\delta}|\langle A, \mathbf{x}_1\mathbf{x}_1^* - \mathbf{x}_2\mathbf{x}_2^*\rangle| &\ \leq \sup_{\mathbf{x}\in S^{n-1}}|\langle A, \mathbf{x}_1\mathbf{x}_1^* - \mathbf{x}_2\mathbf{x}_2^*\rangle| + 4\delta\|A\|\nonumber\\ &\ = (2+4\delta)\|A\|\nonumber\\ &\ = (1+2\delta)\sup_{\mathbf{x}\in S^{n-1}}|\langle A, \mathbf{x}_1\mathbf{x}_1^* - \mathbf{x}_2\mathbf{x}_2^*\rangle|\nn
\end{align}
\end{proof}

\section{Appendix C : Non-convex lanscape}
\subsection{Supporting Lemmas}
The following intermediate results are crucial in proving Theorem \ref{thm:strictsaddle}
\begin{lemma}\label{lem:hermitian_cross}
Let $\{A_d\}_{d=1}^n$ be a set of Hermitian Gaussian random matrices. Then with probability $1-4e^{-cmD(\epsilon)}$,
\begin{align}
    \left| \frac{1}{m}\sum_{d=1}^m\langle A_d, \Delta\bar{\Delta}^\top\rangle\langle A_d, \mathbf{x}\bar{\mathbf{x}}^\top\rangle -  \langle A_d, \Delta\bar{\mathbf{x}}^\top\rangle\langle A_d, \mathbf{x}\bar{\Delta}^\top\rangle \right| \leq \epsilon \nonumber
\end{align}
where \begin{equation}
    D(\epsilon) := \min\left\{\frac{\epsilon^2}{4K_4^2\|\Delta\|_2^4\|\mathbf{x}\|_2^4}, \frac{\epsilon}{2K_4\|\Delta\|_\infty^2\|\mathbf{x}\|_\infty^2}\right\}\nonumber
\end{equation}
\end{lemma}
\begin{proof}
Let $A\in \mathbb{C}^{n\times n}$ be a complex Hermitian Gaussian random matrix, i.e. 
\begin{enumerate}
    \item $\forall~i$, $a_{ii}\sim \mathcal{N}(0, \sigma^2)$.
    \item $\forall~i, j$, $i\neq j$, $a_{ij}\sim \mathcal{N}(0,\frac{\sigma^2}{2}) + i\mathcal{N}(0,\frac{\sigma^2}{2})$.
\end{enumerate}
Define the random variable $Y$,
\begin{align}
    Y&\ = \frac{1}{m}\sum_{d=1}^m\langle A_d, \Delta\bar{\Delta}^\top\rangle\langle A_d, \mathbf{x}\bar{\mathbf{x}}^\top\rangle -  \langle A_d, \Delta\bar{\mathbf{x}}^\top\rangle\langle A_d, \mathbf{x}\bar{\Delta}^\top\rangle \nonumber\\
    &\ = 
    \frac{1}{m}\sum_{d=1}^m \left(\sum_{ij}a_{ij}^d\Delta_i\bar{\Delta}_j\right)\left( \sum_{kl}a_{kl}^d\mathbf{x}_k\bar{\mathbf{x}}_l\right)  - \left(\sum_{ij}a_{ij}^d\mathbf{x}_i\bar{\Delta}_j\right) \left(\sum_{kl}a_{kl}^d\Delta_k\bar{\mathbf{x}}_l\right)\nonumber
\end{align}
For any $A_d$, Split the entire summation $(i,j,k,l) \in [1, n]^4$ into the following 4 sets such that:
\begin{enumerate}
    \item $A:= \{(i,j,k,l) | i = j = k =l\}$
    \item $B:= \{(i,j,k,l) | i=k, j=l\}\cap A^C$
    \item $C:= \{(i,j,k,l) | i=l, j=k\}\cap A^C$
    \item $D:= \{(i,j,k,l)\}\cap A^C \cap B^C \cap C^C$
\end{enumerate}
Calculating the expectation of the sum of the elements in each individual sets:
\begin{enumerate}
    \item For set A,
    \begin{align}
        &\ \mathbb{E}\left[\left(\sum_{ij}a_{ij}\Delta_i\bar{\Delta}_j\right)\left( \sum_{kl}a_{kl}\mathbf{x}_k\bar{\mathbf{x}}_l\right)- \left(\sum_{ij}a_{ij}\mathbf{x}_i\bar{\Delta}_j\right) \left(\sum_{kl}a_{kl}\Delta_k\bar{\mathbf{x}}_l\right)\right]\nonumber\\  &\ = \mathbb{E}\left[a_{ii}^2\Delta_i\bar{\Delta}_i\mathbf{x}_i\bar{\mathbf{x}}_j - a_{ii}^2\Delta_i\bar{\Delta}_i\mathbf{x}_i\bar{\mathbf{x}}_i\right] = \mathbb{E}[0] = 0\nonumber
    \end{align}
    \item For set B,
    \begin{align}
        &\ \mathbb{E}\left[\left(\sum_{ij}a_{ij}\Delta_i\bar{\Delta}_j\right)\left( \sum_{kl}a_{kl}\mathbf{x}_k\bar{\mathbf{x}}_l\right)- \left(\sum_{ij}a_{ij}\mathbf{x}_i\bar{\Delta}_j\right) \left(\sum_{kl}a_{kl}\Delta_k\bar{\mathbf{x}}_l\right)\right]\nonumber\\  &\ = \mathbb{E}\left[a_{ij}^2\Delta_i\bar{\Delta}_j\mathbf{x}_i\bar{\mathbf{x}}_j - a_{ij}^2\Delta_i\bar{\Delta}_j\mathbf{x}_i\bar{\mathbf{x}}_j\right] = \mathbb{E}[0] = 0\nonumber
    \end{align}
    \item For set C, since the matrix $A$ is hermitian $a_{ij}= \bar{a}_{ji}$
    \begin{align}
        &\ \mathbb{E}\left[\left(\sum_{ij}a_{ij}\Delta_i\bar{\Delta}_j\right)\left( \sum_{kl}a_{kl}\mathbf{x}_k\bar{\mathbf{x}}_l\right)- \left(\sum_{ij}a_{ij}\mathbf{x}_i\bar{\Delta}_j\right) \left(\sum_{kl}a_{kl}\Delta_k\bar{\mathbf{x}}_l\right)\right]\nonumber\\  &\ = \left[|a_{ij}|^2\Delta_i\bar{\Delta}_j\mathbf{x}_j\bar{\mathbf{x}}_i - |a_{ij}|^2\Delta_j\bar{\Delta}_j\mathbf{x}_i\bar{\mathbf{x}}_i\right]\nonumber
    \end{align}
    Notice that $\forall i,j$
\begin{align}
    = |a_{ji}|^2\Delta_j\bar{\Delta}_i\mathbf{x}_i\bar{\mathbf{x}}_j +|a_{ij}|^2\Delta_i\bar{\Delta}_j\mathbf{x}_j\bar{\mathbf{x}}_i - |a_{ji}|^2\Delta_i\bar{\Delta}_i\mathbf{x}_j\bar{\mathbf{x}}_j -  |a_{ij}|^2\Delta_j\bar{\Delta}_j\mathbf{x}_i\bar{\mathbf{x}}_i\nonumber
\end{align}
Since $|a_{ij}|^2 = |a_{ji}|^2$. Thus,
\begin{align}
    &\ = |a_{ji}|^2\left[\Delta_j\bar{\Delta}_i\mathbf{x}_i\bar{\mathbf{x}}_j +  \Delta_i\bar{\Delta}_j\mathbf{x}_j\bar{\mathbf{x}}_i - \Delta_i\bar{\Delta}_i\mathbf{x}_j\bar{\mathbf{x}}_j -  \Delta_j\bar{\Delta}_j\mathbf{x}_i\bar{\mathbf{x}}_i\right]\nonumber\\
    &\ = |a_{ji}|^2\left[\Delta_j\bar{\Delta}_i\mathbf{x}_i\bar{\mathbf{x}}_j + \Delta_i\bar{\Delta}_j\mathbf{x}_j\bar{\mathbf{x}}_i - \|\Delta_i\|_2^2\|\mathbf{x}_j\|_2^2 -  \|\Delta_j\|_2^2\|\mathbf{x}_i\|_2^2\right]
    \leq 0\nonumber
\end{align}

\item For set D, as all the elements $(i,j,k,l)\in D$ make $a_{ij}, a_{kl}$ independent of each other, we have,
    \begin{align}
        \mathbb{E}\left(\sum a_{ij}a_{kl}\bar{\Delta}_j\bar{\mathbf{x}}_l\left(\Delta_i\mathbf{x}_k -\Delta_k\mathbf{x}_i\right)\right) = 0\nonumber
    \end{align}
\end{enumerate}
Hence we can conclude,
\begin{equation}
    \mathbb{E}\left[\frac{1}{m}\sum_{d=1}^m\left(\langle A_d, \Delta\bar{\Delta}^\top\rangle\langle A_d, \mathbf{x}\bar{\mathbf{x}}^\top\rangle -  \langle A_d, \Delta\bar{\mathbf{x}}^\top\rangle\langle A_d, \mathbf{x}\bar{\Delta}^\top\rangle\right)\right] = 0\nonumber
\end{equation}
We focus our attention on obtaining concentration bounds. Evaluating the behaviour on the elements in set D,
    \begin{align}
        & \frac{1}{m}\sum_{d=1}^m \left(\sum_{(i,j,k,l)\in D} a_{ij}^d a_{kl}^d \bar{\Delta}_j\bar{\mathbf{x}}_l\left(\Delta_i\mathbf{x}_k -\Delta_k\mathbf{x}_i\right)\right) -\mathbb{E}\left(\frac{1}{m}\sum_{d=1}^m \sum_{{(i,j,k,l)\in D}} a_{ij}^d a_{kl}^d \bar{\Delta}_j\bar{\mathbf{x}}_l\left(\Delta_i\mathbf{x}_k -\Delta_k\mathbf{x}_i\right)\right) \nonumber\\
        & =  \sum_{(i,j,k,l)\in D} a_{ij}^d a_{kl}^d \bar{\Delta}_j\bar{\mathbf{x}}_l\left(\Delta_i\mathbf{x}_k -\Delta_k\mathbf{x}_i\right)\nn
    \end{align}
    We can see that the above is a centered subexponential random variable. Hence using Bernstein type inequality \cite{vershynin2010introduction}, we can say that,
    \begin{align}
        & Pr\left(\left|\frac{1}{m}\sum_{d=1}^m\sum_{(i,j,k,l)\in D} a_{ij}^d a_{kl}^d \bar{\Delta}_j\bar{\mathbf{x}}_l\left(\Delta_i\mathbf{x}_k -\Delta_k\mathbf{x}_i\right)\right|\geq t\right)\nonumber\\& \leq 4\text{exp}\left(-cm\min\left\{\frac{t^2}{4K_4^2\|\Delta\|_2^4\|\mathbf{x}\|_2^4}, \frac{t}{2K_4\|\Delta\|_\infty^2\|\mathbf{x}\|_\infty^2}\right\}\right)\nonumber
    \end{align}
    where $K_4:=\max_{i,j}\{\|\sum_{i,j,k,l\in D}(\bar{a}_{ij}a_{kl}^d)_\mathbb{R}\|_{\psi_1}$, $\|\sum_{i,j,k,l\in D}(\bar{a}_{ij}a_{kl}^d)_\mathbb{C}\|_{\psi_1}\}$ is the subexponential norm.
Thus we can argue that,
\begin{align}
    \mathbb{P}\left(\left|\frac{1}{m}\sum_{d=1}^m\langle A_d, \Delta\bar{\Delta}^\top\rangle\langle A_d, \mathbf{x}\bar{\mathbf{x}}^\top\rangle -  \langle A_d, \Delta\bar{\mathbf{x}}^\top\rangle\langle A_d, \mathbf{x}\bar{\Delta}^\top\rangle \right|\geq t \right)\leq 4\text{exp}^{-cmD(t)}\nonumber
\end{align}
where,
\begin{equation}
    D(t) := \min\left\{\frac{t^2}{4K_4^2\|\Delta\|_2^4\|\mathbf{x}\|_2^4}, \frac{t}{2K_4\|\Delta\|_\infty^2\|\mathbf{x}\|_\infty^2}\right\}\nonumber
\end{equation}
\end{proof}

Throughout the rest of the paper, define $\Delta = \mathbf{x} - e^{i\phi}\mathbf{z}$ such that $\phi = \underset{\theta\in [0,2\pi]}{\arg\min}\|\mathbf{x} - e^{i\theta}\mathbf{z}\|_2$ for any $\mathbf{x}, \mathbf{z}\in \mathbf{C}^n$. 
\begin{lemma}
For any $\mathbf{x}\in \mathbf{C}^n$, 
\begin{equation}
    \|\mathbf{x}\mathbf{x}^*\|_F^2 = \|\mathbf{x}\|_2^4\nonumber
\end{equation}
\end{lemma}
\begin{proof}
\begin{align}
    \|\mathbf{x}\mathbf{x}^*\|_F^2&\ = \sum_{i,j=1}^n |x_i\bar{x}_j|^2
    = \sum_{i,j=1}^n (x_i\bar{x}_j)^*(x_i \bar{x}_j) = \sum_{i,j=1}^n x_j\bar{x}_ix_i \bar{x}_j\nonumber\\
    &\ = \sum_{i,j=1}^n |x_j|^2|x_i|^2= (\sum_{i=1}^n |x_i|^2)^2 = \|\mathbf{x}\|_2^4\nonumber
\end{align}
\end{proof}
\begin{lemma}\label{lem:imzero}
The vectors $\mathbf{x} - e^{i\phi}\mathbf{z}$ and $\mathbf{x} + e^{i\phi}\mathbf{z}$ are such that Im$(\langle \mathbf{x} - e^{i\phi}\mathbf{z}, \mathbf{x} + e^{i\phi}\mathbf{z}\rangle)=0$, where $\phi = \underset{\theta\in [0, 2\pi]}{\arg\min} \|\mathbf{x}- e^{i\theta}\mathbf{z}\|_2$
\end{lemma}
\begin{proof}
We know from Lemma 3.7~\cite{balan2016reconstruction} that $\forall \textbf{x}, \textbf{z} \in\mathbb{C}^n$ $\exists \textbf{u}, \textbf{v}\in \mathbb{C}^n$ such that,
\begin{equation}
    \textbf{x}\textbf{x}^* - \textbf{z}\textbf{z}^* = \textbf{u}\textbf{v}^* + \textbf{v}\textbf{u}^* = [[\textbf{u}, \textbf{v}]]\nonumber
\end{equation}
It can be easily verified that few such pairs $\textbf{u}, \textbf{v}$ are given by,
\begin{equation}
    \textbf{u} = \textbf{x}-e^{i\theta}\textbf{z},~~~\textbf{v} = \textbf{x}+e^{i\theta}\textbf{z},~~~\forall\theta\in[0,2\pi]\nonumber
\end{equation}

Next, we focus on $\langle \mathbf{u}, \mathbf{v}\rangle$. We argue that $\exists \theta$ such that Im($\langle (\mathbf{x}-e^{i\theta}\mathbf{z}), (\mathbf{x}+e^{i\theta}\mathbf{z})\rangle) =0$
To this end consider the following, 
\begin{align}
    \langle \mathbf{u}, \mathbf{v}\rangle &\ =  \mathbf{u}^T\bar{\mathbf{v}}\nonumber\\
    &\ = \left(\textbf{x}-e^{i\theta}\textbf{z}\right)^T\overline{(\textbf{x}+e^{i\theta}\textbf{z})}\nonumber\\
    &\ = \left(\textbf{x}-e^{i\theta}\textbf{z}\right)^T\left(\bar{\textbf{x}}+e^{-i\theta}\bar{\textbf{z}}\right)\nonumber\\
    &\ = \langle \mathbf{x}, \mathbf{x}\rangle - e^{i\theta}\langle \mathbf{z}, \mathbf{x}\rangle  + e^{-i\theta}\langle\mathbf{x}, \mathbf{z}\rangle - \langle \mathbf{z}, \mathbf{z}\rangle\nonumber\\
    &\ = \langle \mathbf{x}, \mathbf{x}\rangle - \langle \mathbf{z}, \mathbf{z}\rangle -2i\text{Im}(e^{i\theta}\langle \mathbf{z}, \mathbf{x}\rangle)\nonumber
\end{align}
Im$(e^{i\theta}\langle \mathbf{z}, \mathbf{x}\rangle)$ can only vanish if $\theta = \omega$ where $\omega \in [0, 2\pi]$ is the angle between the two vectors, i.e. $\omega$ is such that $\langle \mathbf{x}, \mathbf{z}\rangle = e^{i\omega}\|\textbf{x}\|\|\textbf{z}\|$. Next we prove that $\omega = \phi$ where,
\begin{equation}
    \phi = \underset{\theta\in [0, 2\pi]}{\arg\min} \|\mathbf{x}- e^{i\theta}\mathbf{z}\|_2\nonumber
\end{equation}
Consider the following argument,
\begin{align}
\underset{\theta\in [0, 2\pi]}{\arg\min} \|\mathbf{x}- e^{i\theta}\mathbf{z}\|_2^2 &\ = \underset{\theta\in [0, 2\pi]}{\arg\min} \left(\mathbf{x}- e^{i\theta}\mathbf{z}\right)^*\left(\mathbf{x}- e^{i\theta}\mathbf{z}\right)\nonumber\\
&\ = \underset{\theta\in [0, 2\pi]}{\arg\min} \|\mathbf{x}\|^2 + \|\mathbf{z}\|^2 -e^{-i\theta}\mathbf{z}^*\mathbf{x} -e^{i\theta}\mathbf{x}^*\mathbf{z}\nonumber\\
&\ = \underset{\theta\in [0, 2\pi]}{\arg\min} \|\mathbf{x}\|^2 + \|\mathbf{z}\|^2 -2Re(e^{-i\theta}\langle \mathbf{x}, \mathbf{z}\rangle)\nonumber\\
&\ = \|\mathbf{x}\|^2 + \|\mathbf{z}\|^2 -2\underset{\theta\in [0, 2\pi]}{\arg\max} Re(e^{-i\theta}\langle \mathbf{x}, \mathbf{z}\rangle)\nonumber
\end{align}
It can be seen easily that the $\underset{\theta\in [0, 2\pi]}{\arg\max} Re(e^{-i\theta}\langle \mathbf{x}, \mathbf{z}\rangle)$ is achieved when $\theta = \omega$. Thus we have proved that $\phi = \omega$.
\end{proof}
\begin{lemma}
Let $\mathbf{x}, \mathbf{z} \in \mathbb{C}^n$. Then, 
\begin{equation}
    \|(\mathbf{x}-e^{i\phi}\mathbf{z})(\mathbf{x}-e^{i\phi}\mathbf{z})^*\|_F^2 \leq 2\|\mathbf{x}\mathbf{x}^* - \mathbf{z}\mathbf{z}^*\|\nonumber
\end{equation}
where $\phi = \underset{\theta\in [0, 2\pi]}{\arg\min} \|\mathbf{x} - e^{i\theta}\mathbf{z}\|$
\end{lemma}
\begin{proof}
Note,
\begin{align}
\underset{\theta\in [0, 2\pi]}{\arg\min} \|\mathbf{x}- e^{i\theta}\mathbf{z}\|^2 &\ = \underset{\theta\in [0, 2\pi]}{\arg\min} \left(\mathbf{x}- e^{i\theta}\mathbf{z}\right)^*\left(\mathbf{x}- e^{i\theta}\mathbf{z}\right)\nonumber\\
&\ = \underset{\theta\in [0, 2\pi]}{\arg\min} \|\mathbf{x}\|^2 + \|\mathbf{z}\|^2 -e^{-i\theta}\mathbf{z}^*\mathbf{x} -e^{i\theta}\mathbf{x}^*\mathbf{z}\nonumber\\
&\ = \underset{\theta\in [0, 2\pi]}{\arg\min} \|\mathbf{x}\|^2 + \|\mathbf{z}\|^2 -2Re(\langle \mathbf{x}, e^{i\theta}\mathbf{z}\rangle)\nonumber
\end{align}
The minimum can only be achieved at a point where  $Re(\mathbf{x}^*(e^{i\phi}\mathbf{z}) ) \geq 0$.
Further notice that the following relation holds,
\begin{equation}\label{eq:matrix_splits}
    \mathbf{x}\mathbf{x}^* - \mathbf{z}\mathbf{z}^* + \Delta\Delta^* = \mathbf{x}\Delta^* + \Delta\mathbf{x}^*
\end{equation}
Hence we can see that, 
\begin{align}
    \|\mathbf{x}\mathbf{x}^* - \mathbf{z}\mathbf{z}^*\|_F^2 &\ = \|\mathbf{x}\Delta^* + \Delta\mathbf{x}^*-\Delta\Delta^*\|_F^2\nonumber
\end{align}
We know that for any matrix $A$, $\|A\|_F^2 = $Tr$(A^HA)$,
\begin{align}
    \|\mathbf{x}\mathbf{x}^* - \mathbf{z}\mathbf{z}^*\|_F^2 &\ = Tr\left((\mathbf{x}\Delta^* + \Delta\mathbf{x}^*-\Delta\Delta^*)^*(\mathbf{x}\Delta^* + \Delta\mathbf{x}^*-\Delta\Delta^*)\right)\nonumber\\
    &\ = \left(\|\mathbf{x}\Delta^*\|_F^2 + (\langle\mathbf{x}, \Delta\rangle)^2 +  (\langle\Delta, \mathbf{x}\rangle)^2 + \|\Delta\mathbf{x}^*\|_F^2 - 2\langle\mathbf{x}, \Delta\rangle\|\Delta\|_F^2 - 2\langle \Delta, \mathbf{x}\rangle\|\Delta\|_F^2  +\|\Delta\Delta^*\|_F^2 \right)\nn
\end{align}
Note that,
\begin{align}
    & (\langle\mathbf{x}, \Delta\rangle)^2 +  (\langle\Delta, \mathbf{x}\rangle)^2 && = (\langle\mathbf{x}, \Delta\rangle)^2 +  \overline{(\langle\mathbf{x}, \Delta\rangle)^2}  && = 2Re(\langle\mathbf{x}, \Delta\rangle)^2\nonumber\\
    & \langle\mathbf{x}, \Delta\rangle\|\Delta\|_F^2 + \langle \Delta, \mathbf{x}\rangle\|\Delta\|_F^2 && = \langle\mathbf{x}, \Delta\rangle\|\Delta\|_F^2 + \overline{\langle \mathbf{x}, \Delta\rangle}\|\Delta\|_F^2 && = 2Re(\langle\mathbf{x}, \Delta\rangle)\|\Delta\|_F^2\nonumber
\end{align}
Thus we can conclude, 
\begin{align}
\|\mathbf{x}\mathbf{x}^* - \mathbf{z}\mathbf{z}^*\|_F^2 &\ =  \left(2\|\langle\mathbf{x},\Delta\rangle\|_F^2 + 2Re((\langle\mathbf{x}, \Delta\rangle)^2) - 4Re(\langle  \mathbf{x}, \Delta\rangle)\|\Delta\|_F^2  +\|\Delta\Delta^*\|_F^2 \right)\nonumber\\
    &\ = \left(2\mathbf{x}^*\mathbf{x}\Delta^*\Delta + 2Re((\langle\mathbf{x}, \Delta\rangle)^2) - 4Re(\mathbf{x}^*\Delta\Delta^*\Delta)  +\|\Delta\Delta^*\|_F^2 \right)\nonumber
\end{align}
Since $\mathbf{x}^*\mathbf{x}\Delta^*\Delta = \|\langle\mathbf{x},\Delta\rangle\|_F^2$, its a real value.
\begin{align}
    \|\mathbf{x}\mathbf{x}^* - \mathbf{x}^*(\mathbf{x}^*)^*\|_F^2
    &\ = 2\mathbf{x}^*\left(\mathbf{x} - \Delta\right)\Delta^*\Delta + 2Re((\langle\mathbf{x}, \Delta\rangle)^2) - 2Re(\mathbf{x}^*\Delta\Delta^*\Delta)  +\|\Delta\Delta^*\|_F^2 \nonumber\\
    &\ = 2\mathbf{x}^*\left(\mathbf{x} - \Delta\right)\Delta^*\Delta + (\langle\mathbf{x}, \Delta\rangle)^2 +  (\langle\Delta, \mathbf{x}\rangle)^2 - \langle\mathbf{x}, \Delta\rangle\|\Delta\|_F^2 - \langle \Delta, \mathbf{x}\rangle\|\Delta\|_F^2  +\|\Delta\Delta^*\|_F^2 \nonumber\\
    &\ = 2\mathbf{x}^*\left(\mathbf{x} - \Delta\right)\Delta^*\Delta + \left(\langle\mathbf{x}, \Delta\rangle - \frac{1}{2}\langle \Delta, \Delta\rangle\right)^2 +  \left((\langle\Delta, \mathbf{x}\rangle -\frac{1}{2}\langle \Delta, \Delta\rangle\right)^2 +\frac{1}{2}\|\Delta\Delta^*\|_F^2 \nonumber\\
    &\ = 2\mathbf{x}^*\left(\mathbf{x} - \Delta\right)\Delta^*\Delta +  2Re\left((\langle\Delta, \mathbf{x}-\frac{1}{2}\Delta\rangle\right)^2 +\frac{1}{2}\|\Delta\Delta^*\|_F^2 \nonumber\\
    &\ = 2\mathbf{x}^*e^{i\phi}\mathbf{z}\Delta^*\Delta +  \frac{1}{2}Re\left((\langle\Delta, \mathbf{x}+e^{i\phi}\mathbf{z}\rangle\right)^2 +\frac{1}{2}\|\Delta\Delta^*\|_F^2 \nonumber
\end{align}
From Lemma~\ref{lem:imzero}, it can be seen that $Im\left(\langle\Delta, \mathbf{x}+e^{i\phi}\mathbf{z}\rangle\right) = 0$. Hence we can say,
\begin{equation}
    \|\mathbf{x}\mathbf{x}^* - \mathbf{z}\mathbf{z}^*\|_F^2 \geq \frac{1}{2}\|\Delta^*\Delta\|_F^2\nonumber
\end{equation}
\end{proof}

\subsection{Wirtinger Calculus}
\noindent
We use standard arguments from wirtinger calculus~\cite{wang2017generalized} to prove results Theorem~\ref{thm:strictsaddle}. The basic intuition is to look at the $\ell_2$-loss function $f$~\eqref{eq:ell2loss0} as function of two real variables in $\mathbb{R}^n$ rather instead of single complex variable in $\mathbb{C}^n$.
\vspace{1em}
\noindent
This workaround is required for the analysis of real function of complex variables because of notions of complex differentiability and conclusions from  Cauchy-Reimann equations~\cite{psraj}. Hence we map the function $f:\mathbb{C}^n \rightarrow \mathbb{R}$ to $g:\mathbb{R}^n\times \mathbb{R}^n \rightarrow \mathbb{R}$ and instead of analysing the properties of $\nabla^2 f$, we analyse the properties of $\nabla^2 g$.

\vspace{1em}
\noindent
We first introduce the mapped function $g$ and the corresponding expressions for $\nabla g$ and $\nabla^2 g$
\vspace{1em}
\noindent
\begin{align}
    f(\mathbf{x}) = g(\mathbf{x}, \bar{\mathbf{x}})&\ = \frac{1}{m}\sum_{i=1}^m g_i(\mathbf{x}, \bar{\mathbf{x}})\nonumber\\ &\ =  \frac{1}{m}\sum_{i=1}^n |\bar{\mathbf{x}}^\top A_i\mathbf{x} - c_i|^2\nonumber
\end{align}
For the gradient $\nabla g$ we have,
\begin{align}
    \nabla g(\mathbf{x}, \bar{\mathbf{x}}) = \frac{1}{m}\sum_{i=1}^n\begin{bmatrix}
    (\bar{\mathbf{x}}^\top A_i\mathbf{x} - c_i)A_i\mathbf{x}\\
     (\bar{\mathbf{x}}^\top A_i\mathbf{x} - c_i)A_i\bar{\mathbf{x}}
    \end{bmatrix}\nonumber
    \end{align}
For the hessian $\nabla^2 g$, we have
\begin{equation}
    \nabla^2 g(\mathbf{x}, \bar{\mathbf{x}}) = \frac{1}{m}\sum_{i=1}^m\begin{bmatrix}
    (2\bar{\mathbf{x}}^\top A_i\mathbf{x} - c_i)A_i & (A_i\mathbf{x})(A_i\mathbf{x})^\top\\
    (A_i\bar{\mathbf{x}})(A_i\bar{\mathbf{x}})^\top & (2\bar{\mathbf{x}}^\top A_i\mathbf{x} - c_i)A_i
    \end{bmatrix}\nonumber
\end{equation}
The following can be verified easily,
    \begin{align}\label{eq:wirtinger_grad}
    \langle \nabla g(\mathbf{x}), \begin{bmatrix} \Delta\\ \bar{\Delta}\end{bmatrix}\rangle &\ = \frac{1}{m}\sum_{i=1}^m\langle A_i, \mathbf{x}\bar{\mathbf{x}}^\top - \mathbf{z}(\bar{\mathbf{z}})^\top\rangle\langle A_i, \mathbf{x}\bar{\Delta}^\top + \Delta\bar{\mathbf{x}}^\top\rangle\nonumber\\
    &\ = \frac{1}{m}\sum_{i=1}^m\langle A_i, \mathbf{x}\bar{\mathbf{x}}^\top - \mathbf{z}(\bar{\mathbf{z}})^\top\rangle\langle A_i, \mathbf{x}\bar{\mathbf{x}}^\top - \mathbf{z}(\bar{\mathbf{z}}^*)^\top + \Delta\bar{\Delta}^\top\rangle
\end{align}
\begin{align}
    \begin{bmatrix}
    \Delta\\\bar{\Delta}\end{bmatrix}^*\nabla^2 g(\mathbf{x}, \bar{\mathbf{x}})\begin{bmatrix}
    \Delta\\\bar{\Delta}\end{bmatrix} &\ =\frac{1}{m}\sum_{i=1}^m\begin{bmatrix}\Delta\\\bar{\Delta}\end{bmatrix}^*\nabla^2 g_i(\mathbf{x}, \bar{\mathbf{x}})\begin{bmatrix}
    \Delta\\\bar{\Delta}\end{bmatrix} \nonumber\\
    &\ =\frac{1}{m}\sum_{i=1}^m (2\mathbf{x}^\top A_i\bar{\mathbf{x}} - b_i)(\bar{\Delta}^\top A_i\Delta + \Delta^\top A_i\bar{\Delta}) + \left((\Delta^\top A_i\bar{\mathbf{x}})^2 + (\bar{\Delta}A_i\mathbf{x})^2\right)\nonumber
\end{align}
\subsection{Proof of Theorem \ref{thm:strictsaddle}}\label{thm:strictsaddle_proof}
\begin{theorem}
 {Let $\{A_i\}_{i=1}^m$ be a set of complex $n\times n$ Gaussian random matrices, and let $m > Cn$ for some 
constant $C>0$.
Let the scalars $\{c_i\}_{i=1}^m$ characterizing the objective function $f$ of problem (P2) be generated by quadratic measurements of an unknown vector $\mathbf{z}$.
Then, for any given $\xi\in(0,1)$,} there exist positive constants $\beta, \gamma,$ and $\zeta$ such that 
the following statements hold with probability at least $1-\xi$:
\begin{itemize}
    \item [1)] The function f is $(\beta, \zeta, \gamma)$-strict saddle, and
    \item [2)] Every local minimum $\mathbf{w}$ 
    of $f$ satisfies $d(\mathbf{w}, \mathbf{z}) = 0$
\end{itemize}
\end{theorem}

\begin{proof}

Notice that,
\begin{align}
    (\Delta A\bar{\mathbf{x}})^2 + (\bar{\Delta}^\top A\mathbf{x})^2 &\ = \left(\langle A, \mathbf{x}\bar{\Delta}^\top + \Delta\bar{\mathbf{x}}^\top\rangle\right)^2 - 2 (\Delta^\top A\bar{\mathbf{x}})(\bar{\Delta}^\top A\mathbf{x})\nonumber\\
    &\ = \left(\langle A, \mathbf{x}\bar{\mathbf{x}}^\top - \mathbf{z}(\bar{\mathbf{z}})^\top + \Delta\bar{\Delta}^\top\rangle\right)^2- 2 (\langle A, \Delta\bar{\mathbf{x}}^\top\rangle)(\langle A, \mathbf{x}\bar{\Delta}^\top\rangle)\nonumber
\end{align}
Using~\eqref{eq:matrix_splits}, we can reorganize,
\begin{align}
    &\ \begin{bmatrix}
    \Delta\\\bar{\Delta}\end{bmatrix}^*\nabla^2 g_i(\mathbf{x}, \bar{\mathbf{x}})\begin{bmatrix}
    \Delta\\\bar{\Delta}\end{bmatrix} \nonumber\\ &\ = \langle A_i, 2\mathbf{x}\bar{\mathbf{x}}^\top - \mathbf{z}(\bar{\mathbf{z}})^\top\rangle\langle A_i, 2\Delta\bar{\Delta}^\top\rangle + \left(\langle A, \mathbf{x}\bar{\mathbf{x}}^\top - \mathbf{z}(\bar{\mathbf{z}})^\top + \Delta\bar{\Delta}^\top\rangle\right)^2 - 2 (\langle A, \Delta\bar{\mathbf{x}}^\top\rangle)(\langle A, \mathbf{x}\bar{\Delta}^\top\rangle)\nonumber\\
    &\ = 2\left(\langle A_i, 2\mathbf{x}\bar{\mathbf{x}}^\top - \mathbf{z}(\bar{\mathbf{z}})^\top\rangle\langle A_i, \Delta\bar{\Delta}^\top\rangle\right) + \langle A_i, \Delta\bar{\Delta}^\top\rangle\langle A_i, \Delta\bar{\Delta}^\top  + \mathbf{x}\bar{\mathbf{x}}^\top - \mathbf{z}(\bar{\mathbf{z}})^\top\rangle \nonumber\\ &\ + \langle A_i, \mathbf{x}\bar{\mathbf{x}}^\top - \mathbf{z}(\bar{\mathbf{z}})^\top\rangle\langle A_i, \Delta\bar{\Delta}^\top + \mathbf{x}\bar{\mathbf{x}}^\top - \mathbf{z}(\bar{\mathbf{z}})^\top\rangle -  2 (\langle A, \Delta\bar{\mathbf{x}}^\top\rangle)(\langle A, \mathbf{x}\bar{\Delta}^\top\rangle)\nonumber\\
    &\ = 2\langle A_i, \Delta\bar{\Delta}^\top\rangle\langle A_i, \mathbf{x}\bar{\mathbf{x}}^\top\rangle + 2\left(\langle A_i, \mathbf{x}\bar{\mathbf{x}}^\top - \mathbf{z}(\bar{\mathbf{z}})^\top\rangle\langle A_i, \Delta\bar{\Delta}^\top\rangle\right) + \langle A_i, \Delta\bar{\Delta}^\top\rangle\langle A_i, \mathbf{x}\bar{\mathbf{x}}^\top - \mathbf{z}(\bar{\mathbf{z}})^\top\rangle \nonumber\\ &\ + \langle A_i, \Delta\bar{\Delta}^\top\rangle\langle A_i, \Delta\bar{\Delta}^\top\rangle + \langle A_i, \mathbf{x}\bar{\mathbf{x}}^\top - \mathbf{z}(\bar{\mathbf{z}})^\top\rangle\langle A_i, \Delta\bar{\Delta}^\top + \mathbf{x}\bar{\mathbf{x}}^\top - \mathbf{z}(\bar{\mathbf{z}})^\top\rangle -  2 (\langle A, \Delta\bar{\mathbf{x}}^\top\rangle)(\langle A, \mathbf{x}\bar{\Delta}^\top\rangle)\nonumber
    \end{align}
    Adding and subtracting $2\langle A_i, \mathbf{x}\bar{\mathbf{x}}^\top - \mathbf{z}(\bar{\mathbf{z}})^\top\rangle\langle A_i, \mathbf{x}\bar{\mathbf{x}}^\top - \mathbf{z}(\bar{\mathbf{z}})^\top\rangle$, reorganizing, 
    \begin{align}
    &\ \begin{bmatrix}
    \Delta\\\bar{\Delta}\end{bmatrix}^*\nabla^2 g_i(\mathbf{x}, \bar{\mathbf{x}})\begin{bmatrix}
    \Delta\\\bar{\Delta}\end{bmatrix} \nonumber\\
    &\ = 2\langle A_i, \Delta\bar{\Delta}^\top\rangle\langle A_i, \mathbf{x}\bar{\mathbf{x}}^\top\rangle + 2\left(\langle A_i, \mathbf{x}\bar{\mathbf{x}}^\top - \mathbf{z}(\bar{\mathbf{z}})^\top\rangle\langle A_i, \Delta\bar{\Delta}^\top + \mathbf{x}\bar{\mathbf{x}}^\top - \mathbf{z}(\bar{\mathbf{z}})^\top\rangle\right) \nonumber\\
    &\ - 2\langle A_i, \mathbf{x}\bar{\mathbf{x}}^\top - \mathbf{z}(\bar{\mathbf{z}})^\top\rangle\langle A_i, \mathbf{x}\bar{\mathbf{x}}^\top - \mathbf{z}(\bar{\mathbf{z}})^\top\rangle + \langle A_i, \Delta\bar{\Delta}^\top\rangle\langle A_i, \mathbf{x}\bar{\mathbf{x}}^\top - \mathbf{z}(\bar{\mathbf{z}})^\top\rangle\nonumber\\ &\ + \langle A_i, \Delta\bar{\Delta}^\top\rangle\langle A_i, \Delta\bar{\Delta}^\top\rangle + \langle A_i, \mathbf{x}\bar{\mathbf{x}}^\top - \mathbf{z}(\bar{\mathbf{z}})^\top\rangle\langle A_i, \Delta\bar{\Delta}^\top + \mathbf{x}\bar{\mathbf{x}}^\top - \mathbf{z}(\bar{\mathbf{z}})^\top\rangle -  2 (\langle A, \Delta\bar{\mathbf{x}}^\top\rangle)(\langle A, \mathbf{x}\bar{\Delta}^\top\rangle)\nonumber
    \end{align}
    Adding and subtracting $\langle A_i, \mathbf{x}\bar{\mathbf{x}}^\top - \mathbf{z}(\bar{\mathbf{z}})^\top\rangle\langle A_i, \mathbf{x}\bar{\mathbf{x}}^\top - \mathbf{z}(\bar{\mathbf{z}})^\top\rangle$, reorganizing, 
    \begin{align}
    &\ \begin{bmatrix}
    \Delta\\\bar{\Delta}\end{bmatrix}^*\nabla^2 g_i(\mathbf{x}, \bar{\mathbf{x}})\begin{bmatrix}
    \Delta\\\bar{\Delta}\end{bmatrix} \nonumber\\
    &\ = 2\langle A_i, \Delta\bar{\Delta}^\top\rangle\langle A_i, \mathbf{x}\bar{\mathbf{x}}^\top\rangle -  2 (\langle A, \Delta\bar{\mathbf{x}}^\top\rangle)(\langle A, \mathbf{x}\bar{\Delta}^\top\rangle) - 3\langle A_i, \mathbf{x}\bar{\mathbf{x}}^\top - \mathbf{z}(\bar{\mathbf{z}})^\top\rangle\langle A_i, \mathbf{x}\bar{\mathbf{x}}^\top - \mathbf{z}(\bar{\mathbf{z}})^\top\rangle\nonumber\\
    &\ + \langle A_i, \Delta\bar{\Delta}^\top\rangle\langle A_i, \Delta\bar{\Delta}^\top\rangle + 4\langle A_i, \mathbf{x}\bar{\mathbf{x}}^\top - \mathbf{z}(\bar{\mathbf{z}})^\top\rangle\langle A_i, \Delta\bar{\Delta}^\top + \mathbf{x}\bar{\mathbf{x}}^\top - \mathbf{z}(\bar{\mathbf{z}})^\top\rangle\nonumber
    \end{align}
    Using equation~\eqref{eq:wirtinger_grad}, 
    \begin{align}
    &\ \begin{bmatrix}
    \Delta\\\bar{\Delta}\end{bmatrix}^*\nabla^2 g_i(\mathbf{x}, \bar{\mathbf{x}})\begin{bmatrix}
    \Delta\\\bar{\Delta}\end{bmatrix} \nonumber\\
    &\ = 2\langle A_i, \Delta\bar{\Delta}^\top\rangle\langle A_i, \mathbf{x}\bar{\mathbf{x}}^\top\rangle -  2 (\langle A, \Delta\bar{\mathbf{x}}^\top\rangle)(\langle A, \mathbf{x}\bar{\Delta}^\top\rangle) - 3\langle A_i, \mathbf{x}\bar{\mathbf{x}}^\top - \mathbf{z}(\bar{\mathbf{z}})^\top\rangle\langle A_i, \mathbf{x}\bar{\mathbf{x}}^\top - \mathbf{z}(\bar{\mathbf{z}})^\top\rangle\nonumber\\
    &\ + \langle A_i, \Delta\bar{\Delta}^\top\rangle\langle A_i, \Delta\bar{\Delta}^\top\rangle + 4\langle \nabla g_i(\mathbf{x}, \bar{\mathbf{x}}), \begin{bmatrix}\Delta \\ \bar{\Delta}\end{bmatrix}\rangle\nonumber
\end{align}

Overall, we can conclude that,
\begin{align}
    &\ \begin{bmatrix}
    \Delta\\\bar{\Delta}\end{bmatrix}^*\nabla^2 g(\mathbf{x}, \bar{\mathbf{x}})\begin{bmatrix}
    \Delta\\\bar{\Delta}\end{bmatrix} = \frac{1}{m}\sum_{i=1}^m\begin{bmatrix}
    \Delta\\\bar{\Delta}\end{bmatrix}^*\nabla^2 g_i(\mathbf{x}, \bar{\mathbf{x}})\begin{bmatrix}
    \Delta\\\bar{\Delta}\end{bmatrix}\nonumber\\ &\ = \frac{2}{m}\sum_{i=1}^m \langle A_i, \Delta\bar{\Delta}^\top\rangle\langle A_i, \mathbf{x}\bar{\mathbf{x}}^\top\rangle - (\langle A, \Delta\bar{\mathbf{x}}^\top\rangle)(\langle A, \mathbf{x}\bar{\Delta}^\top\rangle) - \frac{3}{m}\sum_{i=1}^m\langle A_i, \mathbf{x}\bar{\mathbf{x}}^\top - \mathbf{z}(\bar{\mathbf{z}})^\top\rangle\langle A_i, \mathbf{x}\bar{\mathbf{x}}^\top - \mathbf{z}(\bar{\mathbf{z}})^\top\rangle\nonumber\\
    &\ + \frac{1}{m}\sum_{i=1}^m\langle A_i, \Delta\bar{\Delta}^\top\rangle\langle A_i, \Delta\bar{\Delta}^\top\rangle + \frac{4}{m}\sum_{i=1}^m\langle \nabla g_i (\mathbf{x}, \bar{\mathbf{x}}), \begin{bmatrix}\Delta \\ \bar{\Delta}\end{bmatrix}\rangle\nonumber
    \end{align}
    Using Lemma~\ref{lem:subexpo_fro} and Lemma~\ref{lem:hermitian_cross}, we can conclude that with probability greater than $1-c_1e^{-c_2m\min\{ D(t), E(\epsilon)\} }$,
    \begin{align}\label{eq:hessain_curvature}
    \begin{bmatrix}
    \Delta\\\bar{\Delta}\end{bmatrix}^*\nabla^2 g(\mathbf{x}, \bar{\mathbf{x}})\begin{bmatrix}
    \Delta\\\bar{\Delta}\end{bmatrix} &\ \leq  t + 4\delta\|\Delta\|_2 + \beta \|\Delta\Delta^*\|_F^2 - 3\alpha \|\mathbf{x}\mathbf{x}^* - \mathbf{z}(\mathbf{z})^*\|_F^2\nonumber\\
    &\ \leq t + 4\delta\|\Delta\|_2 + 2\beta \|\mathbf{x}\mathbf{x}^* - \mathbf{z}(\mathbf{z})^*\|_F^2 - 3\alpha \|\mathbf{x}\mathbf{x}^* - \mathbf{z}(\mathbf{z})^*\|_F^2\nonumber\\
    &\ \leq \left(2\beta  - 3\alpha \right)\|\mathbf{x}\mathbf{x}^* - \mathbf{z}(\mathbf{z})^*\|_F^2 +4\delta\|\Delta\|_2 + t
\end{align}
where $\exists c_1, c_2>0$ which can be computed from Lemma~\ref{lem:subexpo_fro} and Lemma~\ref{lem:hermitian_cross}.

For any $\xi\in [0,1]$, there can be multiple possibilities of the constants $\beta, \zeta$ and $\gamma$ satisfying Theorem~\ref{thm:strictsaddle}.

Given $\xi$, we can bound $m = O(n)$ such that the mapping $\mathcal{M}_\mathcal{A}$ is $(1-c, 1+c)$-stable, for some small $c >0$ and if the current vector $\mathbf{x}$ is not close to $\mathbf{x}^*$ such that $\|\Delta\|\geq C_0\delta$, for sufficiently large $C_0 > 0$, then we have 
\begin{align}
    \begin{bmatrix}
    \Delta\\\bar{\Delta}\end{bmatrix}^*\nabla^2 g(\mathbf{x}, \bar{\mathbf{x}})\begin{bmatrix}
    \Delta\\\bar{\Delta}\end{bmatrix} \leq (-1 + 5c)C_0^2\delta^2 + 4C_0\delta^2 \leq 0\nonumber
\end{align}

A particular set of $(c, C_0)$ which fit the above condition is $c = \frac{1}{20}, C_0 = 10$, then 
\begin{align}
    \begin{bmatrix}
    \Delta\\\bar{\Delta}\end{bmatrix}^*\nabla^2 g(\mathbf{x}, \bar{\mathbf{x}})\begin{bmatrix}
    \Delta\\\bar{\Delta}\end{bmatrix} \leq (-1 + 5c)C_0^2\delta^2 + 4C_0\delta^2 \leq 0\nonumber
\end{align}
Hence we can conclude that the function $f(\mathbf{x}) = g(\mathbf{x}, \overline{\mathbf{x}})$ satisfies at-least one of the following is true,
\begin{itemize}
    \item $\|\nabla g(\mathbf{x})\|\geq \delta$
    \item For the direction vector $\Delta$, \begin{align}\begin{bmatrix}
    \Delta\\\bar{\Delta}\end{bmatrix}^*\nabla^2 g(\mathbf{x}, \bar{\mathbf{x}})\begin{bmatrix}
    \Delta\\\bar{\Delta}\end{bmatrix} \leq (-1 + 5c)C_0^2\delta^2 + 4C_0\delta^2\nonumber
    \end{align}
    \item $d(\mathbf{x}, \mathbf{z}) \leq C_0\delta$
\end{itemize}
Thus there exists constants $\beta, \zeta,\gamma > 0$ such that the function $f(\mathbf{x})$ is $(\beta, \zeta, \gamma)$ strict saddle.

Following up on equation \eqref{eq:hessain_curvature}, the only possible way that the hessain $\nabla^2 g(\mathbf{x}, \overline{\mathbf{x}}) \succeq 0$ is if $\|\mathbf{x}\mathbf{x}^*-\mathbf{z}\mathbf{z}^*\| = 0$. Hence we can conclude that all local minimas, i.e. all $\mathbf{w}$ such that $\nabla^2 g(\mathbf{w}, \overline{\mathbf{w}}) \succeq 0$ has to satisfy $\|\mathbf{w}\mathbf{w}^*-\mathbf{z}\mathbf{z}^*\| =0$ and hence satisfies $\mathbf{z}\sim\mathbf{w}$ which makes $\mathbf{w}$ the solution of the problem \eqref{eq:main_prob}.

\end{proof}

\section{Appendix D : Applications}\label{app:application}
\subsection{Power system state estimation problem}
 
Apart from being a broader class of problems encompassing phase retrieval, the problem setup \eqref{eq:main_prob} also has applications in power system engineering. Given a network of buses and transmission lines, the goal is to estimate complex voltages across all buses from a subset of noisy power and voltage magnitude measurements. In the AC power model, these measurements are quadratically dependent on the voltage values to be determined. Let $\{c_i\}_{i=1}^m$ be the set of measurements and $\{A_i\}_{i=1}^m$ be the corresponding bus admittance value matrices. Then the problem boils down to an estimation problem 
\begin{align}
    \text{find }&\ \mathbf{x}\nonumber\\
    \text{s.t. }&\ c_i = \mathbf{x}^*A_i\mathbf{x} + \nu_i~~~\forall i =1,2,\dots ,m.\nonumber
\end{align}
where $\nu_i \sim \mathcal{N}(0, \sigma_i^2)$ is gaussian noise associated with the readings.\cite{konar2017first}. For details on the problem setup, please refer \cite{wang2016power}.

\subsection{Fusion Phase retrieval}
 
Let $\{W_i\}_{i=1}^m$ be a set of subspace of $\mathbb{R}^n/\mathbb{C}^n$. Fusion phase retrieval deals with the problem of recovering $\mathbf{x}$ upto a phase ambiguity from the measurements of the form $\{\|P_i\mathbf{x}\|\}_{i=1}^m$, where $P_i:\mathbb{C}^n/\mathbb{R}^n\rightarrow W_i$ are projection operators onto the subspaces. 
\cite{cahill2013phase} had the initial results on this problem with regards to the conditions on the subspaces and minimum number of such subspaces required for successful recovery of $\mathbf{x}$ under phase ambiguity.
\subsection{X-ray crystallography}
 
In X-ray crystallography, especially in crystal twinning \cite{drenth2007principles}, the measurements are obtained with orthogonal matrices $Q_i^2 = Q_i$ which again would be solved by out setup.
 
In the worst case, a feasibility quadratic feasibility problem can be NP-hard, which makes the setup \eqref{eq:main_prob} we address all the more interesting as we can highlight properties about a subgroup of quadratic feasibility problems and take a shot at providing provably converging algorithm for the same. This question resonates quite closely with many applications of quadratic feasibility as discussed above. In this write-up we have only considered the noiseless system, which later can be extended to noisy system analysis.

\end{document}